\documentclass[11pt]{article}%
\usepackage{amsmath}
\usepackage{amsfonts}
\usepackage{amssymb}
\usepackage{graphicx}%
\newtheorem{theorem}{Theorem}

\newtheorem{corollary}[theorem]{Corollary}

\newtheorem{definition}[theorem]{Definition}

\newtheorem{lemma}[theorem]{Lemma}

\newtheorem{proposition}[theorem]{Proposition}
\newtheorem{remark}[theorem]{Remark}

\newenvironment{proof}[1][Proof]{\noindent\textbf{#1.} }{\ \rule{0.5em}{0.5em}}
\begin{document}

\title{The quantum harmonic oscillator with icosahedral symmetry and some explicit wavefunctions}
\author{Charles F. Dunkl\\Department of Mathematics\\University of Virginia\\Charlottesville VA 22904-4137}
\maketitle

\begin{abstract}
The Dunkl Laplacian is used to define the Hamiltonian of a modified quantum
harmonic oscillator, associated with any finite reflection group. The
potential is a sum of the inverse squares of the linear functions whose zero
sets are the mirrors of the group's reflections. The symmetric group version
of this is known as the Calogero-Moser model of N identical particles on a
line. This paper focuses on the group of symmetries of the regular
icosahedron, associated to the root system of type H3. Special wavefunctions
are defined by a generating function arising from the vertices of the
icosahedron and have the key property of allowing easy calculation of the
effect of the Dunkl Laplacian. The ground state is the product of a Gaussian
function with powers of linear functions coming from the root system. Two
types of wavefunctions are considered, inhomogeneous polynomials with
specified top-degree part, and homogeneous harmonic polynomials. The squared
norms for both types are explicitly calculated. Symmetrization is applied to
produce the invariant polynomials of both types, as well as their squared
norms. The action of the angular momentum square on the harmonic homogeneous
polynomials is determined. There is also a sixth-order operator commuting with
the Hamiltonian and the group action.

\end{abstract}

\section{Introduction}

The icosahedron appears in geometry as one of the five regular solids, in
algebra as one of the rank three Coxeter groups and even in nature as a virus
molecule. This paper studies its role in quantum mechanics, in the form of a
potential with icosahedral symmetry in the Schr\"{o}dinger equation. Just as
in antiquity there are always mysteries about the icosahedron but some new
analytical insights will be presented here. We use the mechanism of Dunkl
operators and the Dunkl Laplacian to define the Hamiltonian of a modified
quantum harmonic oscillator, associated with any finite reflection group. The
potential is a sum of the inverse squares of the linear functions whose zero
sets are the mirrors of the group's reflections. The symmetric group version
of this is known as the Calogero-Moser model of $N$ identical particles on a
line with $r^{-2}$ interactions and harmonic confinement (Lapointe and Vinet
\cite{LV1996}). Quesne \cite{Q2010} proved superintegrability for the even
dihedral group oscillator. The abelian group cases (sign-changes of the
coordinates) have been studied by Genest and Vinet \cite{GV2014}, who used the
Hamiltonian to define several symmetry algebras.

Section \ref{RefGrps} provides the general background for the paper and
applies to all finite reflection groups. In Section \ref{Rootsys} there is
background information about finite reflection groups, Dunkl operators and the
Dunkl Laplacian. There are some formulas for harmonic polynomials, that is,
polynomials annihilated by the Dunkl Laplacian, which are needed later. Also
there are commutation relations for multiplication and Dunkl operators.
Section \ref{DHO} introduces the Dunkl harmonic oscillator and basic
properties, for any reflection group. The interplay between wavefunctions,
harmonic and Laguerre polynomials is described. The definition and properties
of important operators, namely the raising (creation) and lowering
(annihilation) types and the angular momentum squares, are presented. The
ground states for the Schr\"{o}dinger equation are products of an exponential
and a group-invariant product function formed from the root system. This leads
to the inner product structure which underlies the probabilities coming from
squares of wavefunctions. The formulas in this part apply in the general setting.

The icosahedral group and its root system $H_{3}$ are introduced in Section
\ref{IcosGrp}. There is a brief overview of the symmetry properties of the
regular icosahedron and the regular dodecahedron, and of the fundamental
invariant polynomials. Also the $H_{3}$ version of the Macdonald integral
appears; this serves to normalize the squared-norm formulas for wavefunctions.
The special wavefunctions of the title are defined in Section \ref{GenFunWF}.
They are produced by a generating function arising from the vertices of the
icosahedron and have the key property of allowing easy calculation of the
effect of the Dunkl Laplacian. Two types of wavefunctions are considered,
inhomogeneous polynomials with specified top-degree part, and homogeneous
harmonic polynomials. The squared norms for both types are explicitly
calculated. Symmetrization is applied to produce the invariant polynomials of
both types, as well as the squared norms. In Section \ref{2nd6thOps} two other
operators are discussed, the angular momentum square and a sixth-order
operator defined in terms of raising and lowering operators. There are
concluding remarks (Section \ref{ConcRem}) regarding some open technical
problems in the construction and analysis of wavefunctions for the icosahedral
model. The Appendix sketches a symbolic computation approach to analyzing some
operators defined in terms of Dunkl operators.

\section{\label{RefGrps}Reflection groups and background}

\subsection{\label{Rootsys}Root systems and Dunkl operators}

In $\mathbb{R}^{N}$ the inner product is $\left\langle x,y\right\rangle
:=\sum_{i=1}^{N}x_{i}y_{i}$ and $\left\vert x\right\vert ^{2}:=\left\langle
x,x\right\rangle $. If $v\neq0$ then the reflection $\sigma_{v}$ along $v$ is
defined by%
\[
x\sigma_{v}:=x-2\frac{\left\langle x,v\right\rangle }{\left\vert v\right\vert
^{2}}v.
\]
This is an isometry $\left\vert x\sigma_{v}\right\vert ^{2}=\left\vert
x\right\vert ^{2}$ and an involution $\sigma_{v}^{2}=I$. The set of fixed
points ($x\sigma_{v}=x$) is the hyperplane $\left\{  x:\left\langle
x,v\right\rangle =0\right\}  $. A finite root system is a subset $R$ of
nonzero elements of $\mathbb{R}^{N}$ satisfying $u,v\in R$ implies
$u\sigma_{v}\in R$. We use only reduced root systems, that is, if $u,cu\in R$
then $c=\pm1$. Define $W\left(  R\right)  $ to be the (\textit{finite
reflection}) group generated by $\left\{  \sigma_{v}:v\in R\right\}  $, a
finite subgroup of the orthogonal group $O_{N}\left(  \mathbb{R}\right)  .$
There is a decomposition of $R$ into positive roots $R_{+}=\left\{  v\in
R:\left\langle u_{0},v\right\rangle >0\right\}  $ and $R_{-}$; where $u_{0}$
is some fixed vector such that $\left\langle u_{0},v\right\rangle \neq0$ for
all $v\in R$. Since $\sigma_{v}=\sigma_{-v}$ the set $R_{+}$ is used to index
the reflections in $W\left(  R\right)  $. The root system $R$ is a union of
conjugacy classes $(W\left(  R\right)  $ orbits): $\sigma_{u}\sim\sigma_{v}$
if $u=vw$ for some $w\in W\left(  R\right)  .$ A \textit{multiplicity
function} $\kappa_{v}$ is a function on $R$ which is constant on each
conjugacy class, and usually here $\kappa_{v}\geq1$ or $\kappa_{v}$ is a
formal parameter. Set $\gamma_{\kappa}:=\sum_{v\in R_{+}}\kappa_{v}$. Define
the Dunkl operator ($1\leq i\leq N$)%
\[
\mathcal{D}_{i}f\left(  x\right)  :=\frac{\partial}{\partial x_{i}}f\left(
x\right)  +\sum\limits_{v\in R_{+}}\kappa_{v}\frac{f\left(  x\right)
-f\left(  x\sigma_{v}\right)  }{\left\langle x,v\right\rangle }v_{i}.
\]
Then $\mathcal{D}_{i}\mathcal{D}_{j}=\mathcal{D}_{j}\mathcal{D}_{i}$ for all
$i,j$ (Dunkl \cite{D1989}, also see Dunkl and Xu \cite[Sect. 6.4]{DX2014}).
Let $\nabla=\left(  \frac{\partial}{\partial x_{1}},\cdots,\frac{\partial
}{\partial x_{N}}\right)  $, $\Delta=\sum_{i=1}^{N}\left(  \frac{\partial
}{\partial x_{i}}\right)  ^{2}$ and $\nabla_{\kappa}=\left(  \mathcal{D}%
_{1},\ldots,\mathcal{D}_{N}\right)  .$

The Dunkl Laplacian is $\Delta_{\kappa}:=\sum_{i=1}^{N}\mathcal{D}_{i}^{2}$
and
\[
\Delta_{\kappa}f\left(  x\right)  =\Delta f\left(  x\right)  +\sum_{v\in
R_{+}}\kappa_{v}\left(  2\frac{\left\langle \nabla f\left(  x\right)
,v\right\rangle }{\left\langle x,v\right\rangle }-\left\vert v\right\vert
^{2}\frac{f\left(  x\right)  -f\left(  x\sigma_{v}\right)  }{\left\langle
x,v\right\rangle ^{2}}\right)  .
\]
Let $\mathbb{F}$ denote an extension field of $\mathbb{R}$ containing
$\omega>0$ and the values of $\kappa_{v}$ ($\mathbb{C}$ is not used here). Set
$\mathcal{P}:=\mathbb{F}\left[  x_{1},\ldots,x_{N}\right]  $ (polynomials in
$x$), for $n=1,2,\ldots$ let $\mathcal{P}_{n}=\left\{  p\in\mathcal{P}%
:p\left(  cx\right)  =c^{n}p\left(  x\right)  ,c\in\mathbb{F}\right\}  $
(homogeneous polynomials) and $\mathcal{P}_{n,\kappa}=\left\{  p\in
\mathcal{P}_{n}:\Delta_{\kappa}p\left(  x\right)  =0\right\}  $, the harmonic
homogeneous polynomials. The group $W\left(  R\right)  $ is represented on
$\mathcal{P}$ by $wp\left(  x\right)  =p\left(  xw\right)  $, $w\in W\left(
R\right)  $.

There are some basic commutation relations used throughout. (Note $\left[
A,B\right]  :=AB-BA$ for operators, and $\left\langle a,x\right\rangle
,\left\vert x\right\vert ^{2}$ are interpreted as multiplication operators,
for $a\in\mathbb{R}$):%
\begin{equation}
\left[  \left\langle a,\nabla_{\kappa}\right\rangle ,\left\langle
b,x\right\rangle \right]  =\left\langle a,b\right\rangle +2\sum_{v\in R_{+}%
}\kappa_{v}\frac{\left\langle a,v\right\rangle \left\langle b,v\right\rangle
}{\left\vert v\right\vert ^{2}}\sigma_{v} \label{xDx}%
\end{equation}%
\begin{equation}
\left[  \Delta_{\kappa},\left\langle b,x\right\rangle \right]  =2\left\langle
b,\nabla_{\kappa}\right\rangle ,~\left[  \left\vert x\right\vert
^{2},\left\langle a,\nabla_{\kappa}\right\rangle \right]  =-2\left\langle
a,x\right\rangle . \label{Dbxsq}%
\end{equation}
Denote $\delta:=\left\langle x,\nabla\right\rangle $ (thus $\delta p=np$ for
$p\in\mathcal{P}_{n}$), then%
\begin{equation}
\left[  \Delta_{\kappa},\left\vert x\right\vert ^{2}\right]  =2\left(
N+2\gamma_{\kappa}+2\delta\right)  . \label{[xsq,delta]}%
\end{equation}

\begin{proposition}
\label{Lap^m}Suppose $\phi\in\mathcal{P}_{n,\kappa}$ and $1\leq k\leq
m=1,2,3,\ldots$ then%
\begin{align}
\Delta_{\kappa}\left\vert x\right\vert ^{2k}\phi\left(  x\right)   &
=2k\left(  N+2\gamma_{\kappa}+2n+2k-2\right)  \left\vert x\right\vert
^{2k-2}\phi\left(  x\right) \label{delxsq1}\\
\Delta_{\kappa}^{m}\left\vert x\right\vert ^{2m}\phi\left(  x\right)   &
=2^{2m}m!\left(  \frac{N}{2}+\gamma_{\kappa}+n\right)  _{m}~\phi\left(
x\right)  .\nonumber
\end{align}

\end{proposition}

\begin{proof}
By (\ref{[xsq,delta]}), which is the case $i=1$, and induction we prove the
first statement. Assume the statement is true for $i=1,2,\ldots,k-1$ then
(with $f=\left\vert x\right\vert ^{2k-2}\phi$)%
\begin{gather*}
\Delta_{\kappa}\left\vert x\right\vert ^{2k}\phi\left(  x\right)  =\left\vert
x\right\vert ^{2}\Delta_{\kappa}\left\vert x\right\vert ^{2k-2}\phi+2\left(
N+2\gamma_{\kappa}+2\delta\right)  \left\vert x\right\vert ^{2k-2}\phi\left(
x\right) \\
=\left\{  2\left(  k-1\right)  \left(  N+2\gamma_{\kappa}+2n+2k-4\right)
+2\left(  N+2\gamma_{\kappa}+4k-4+2n\right)  \right\}  \left\vert x\right\vert
^{2k-2}\phi\\
=2k\left(  N+2\gamma_{k}+2n+2k-2\right)  \left\vert x\right\vert ^{2k-2}\phi,
\end{gather*}
completing the induction. Let $C_{n,k}=4k\left(  \frac{N}{2}+\gamma_{\kappa
}+n+k-1\right)  $ for $k=1,2,\ldots,m$ then
\begin{align*}
\Delta_{\kappa}^{m-1}\Delta_{\kappa}\left\vert x\right\vert ^{2m}\phi\left(
x\right)   &  =C_{n,m}\Delta_{\kappa}^{m-2}\left(  \Delta_{\kappa}\left\vert
x\right\vert ^{2m-2}\phi\left(  x\right)  \right) \\
&  =C_{n,m}C_{n,m-1}\Delta_{\kappa}^{m-3}\left(  \Delta_{\kappa}\left\vert
x\right\vert ^{2m-4}\phi\left(  x\right)  \right)  =\ldots\\
&  =C_{n,m}C_{n,m-1}\cdots C_{n,2}\Delta_{\kappa}\left\vert x\right\vert
^{2}\phi\left(  x\right)  =\prod\limits_{i=1}^{m}C_{n,i}\phi\left(  x\right)
\end{align*}
and $\prod\limits_{i=1}^{m}C_{n,i}=2^{2m}m!\left(  \frac{N}{2}+\gamma_{\kappa
}+n\right)  _{m}$.
\end{proof}

There is a direct sum decomposition $\mathcal{P}_{n}=\sum_{j=0}^{\left\lfloor
n/2\right\rfloor }\oplus\left\vert x\right\vert ^{2j}\mathcal{P}_{n-2j,\kappa
}$; this is a consequence of the following two propositions (\cite[Thm.
7.1.15]{DX2014}) 

\begin{proposition}
\label{hmproj}For $n=1,2,\ldots$ set $\Lambda_{n}:=\sum\limits_{j=0}%
^{\left\lfloor n/2\right\rfloor }\dfrac{1}{4^{j}j!\left(  -N/2-\gamma_{\kappa
}-n+2\right)  _{j}}\left\vert x\right\vert ^{2j}\Delta_{\kappa}^{j}$ then
$p\in\mathcal{P}_{n}$ implies .$\Lambda_{n}p\in\mathcal{P}_{n,\kappa}$ 
\end{proposition}

\begin{proposition}
\label{hmex}Suppose $p\in\mathcal{P}_{n}$ then%
\[
p\left(  x\right)  =\sum_{j=0}^{\left\lfloor n/2\right\rfloor }\dfrac{1}%
{4^{j}j!\left(  N/2+\gamma_{\kappa}+n-2j\right)  _{j}}\left\vert x\right\vert
^{2j}\Lambda_{n-2j}\Delta_{\kappa}^{j}p\left(  x\right)  .
\]

\end{proposition}

\begin{definition}
For $a,b\in\mathbb{R}^{N}$ the angular momentum operator is $J_{a,b}%
:=\left\langle a,x\right\rangle \left\langle b,\nabla_{\kappa}\right\rangle
-\left\langle b,x\right\rangle \left\langle a,\nabla_{\kappa}\right\rangle .$
\end{definition}

\begin{proposition}
\label{Jprops}$J_{a,b}=\left\langle b,\nabla_{\kappa}\right\rangle
\left\langle a,x\right\rangle -\left\langle a,\nabla_{\kappa}\right\rangle
\left\langle b,x\right\rangle $; and $\left[  \Delta_{\kappa},J_{a,b}\right]
=0=\left[  \left\vert x\right\vert ^{2},J_{a,b}\right]  $.
\end{proposition}

\begin{proof}
From (\ref{xDx}) the commutator $\left[  \left\langle a,x\right\rangle
,\left\langle b,\nabla_{\kappa}\right\rangle \right]  =-\left[  \left\langle
b,x\right\rangle ,\left\langle a,\nabla_{\kappa}\right\rangle \right]  $, and
this proves the first statement. Next (by \ref{Dbxsq})%
\begin{align*}
\left[  J_{a,b},\left\vert x\right\vert ^{2}\right]   &  =2\left\langle
a,x\right\rangle \left\langle b,x\right\rangle -2\left\langle b,x\right\rangle
\left\langle a,x\right\rangle =0,\\
J_{a,b}\Delta_{\kappa} &  =\left\langle a,x\right\rangle \Delta_{\kappa
}\left\langle b,\nabla_{\kappa}\right\rangle -\left\langle b,x\right\rangle
\Delta_{\kappa}\left\langle a,\nabla_{\kappa}\right\rangle \\
&  =\Delta_{\kappa}J_{a,b}+\left(  -2\left\langle a,\nabla_{\kappa
}\right\rangle \left\langle b,\nabla_{\kappa}\right\rangle +2\left\langle
b,\nabla_{\kappa}\right\rangle \left\langle a,\nabla_{\kappa}\right\rangle
\right)  =\Delta_{\kappa}J_{a,b}.
\end{align*}

\end{proof}

This family of angular momentum operators has been analyzed by Feigin and
Hakobyan \cite{FH2015}, with emphasis on the symmetric group and the
Calogero-Moser model.

The standard unit basis vectors are denoted $\varepsilon_{i}=\left(
0\ldots,\overset{i}{1},0\ldots\right)  $ for $1\leq i\leq N$.

\begin{definition}
The angular momentum square is $\mathcal{J}:=\sum\limits_{1\leq i<j\leq
N}J_{\varepsilon_{i},\varepsilon_{j}}^{2}$
\end{definition}

\begin{theorem}
\label{angsq}$\mathcal{J}=\left\vert x\right\vert ^{2}\Delta_{\kappa
}-\left\langle x,\nabla_{\kappa}\right\rangle ^{2}-\left(  N-2\right)
\left\langle x,\nabla_{\kappa}\right\rangle -2\sum\limits_{v\in R_{+}}%
\kappa_{v}\sigma_{v}\left\langle x,\nabla_{\kappa}\right\rangle $.
\end{theorem}

\begin{proof}
Recall the commutation $\mathcal{D}_{i}x_{j}=x_{j}\mathcal{D}_{i}+\delta
_{ij}+2\alpha_{ij}$ with $\alpha_{ij}:=\sum\limits_{v\in R_{=}}\kappa
_{v}\dfrac{v_{i}v_{j}}{\left\vert v\right\vert ^{2}}\sigma_{v}$. First%
\begin{align*}
2\sum\limits_{1\leq i<j\leq N}J_{\varepsilon_{i},\varepsilon_{j}}^{2}  &
=\sum_{i\neq j}\left(  x_{i}\mathcal{D}_{j}x_{i}\mathcal{D}_{j}+x_{j}%
\mathcal{D}_{i}x_{j}\mathcal{D}_{i}-x_{i}\mathcal{D}_{j}x_{j}\mathcal{D}%
_{i}-x_{j}\mathcal{D}_{i}x_{i}\mathcal{D}_{j}\right) \\
&  =\sum_{i\neq j}\left(
\begin{array}
[c]{c}%
x_{i}^{2}\mathcal{D}_{j}^{2}+x_{j}^{2}\mathcal{D}_{i}^{2}+2x_{i}\alpha
_{ij}\mathcal{D}_{j}+2x_{j}\alpha_{ij}\mathcal{D}_{i}-2x_{i}x_{j}%
\mathcal{D}_{i}\mathcal{D}_{j}\\
-x_{i}\mathcal{D}_{i}-x_{j}\mathcal{D}_{j}-2x_{i}\alpha_{jj}\mathcal{D}%
_{i}-2x_{j}\alpha_{ii}\mathcal{D}_{j}%
\end{array}
\right)
\end{align*}%
\begin{align*}
\left\langle x,\nabla_{\kappa}\right\rangle ^{2}  &  =\sum_{i,j=1}^{N}%
x_{i}\mathcal{D}_{i}x_{j}\mathcal{D}_{j}=\sum_{i=1}^{N}\left(  x_{i}%
^{2}\mathcal{D}_{i}^{2}+x_{i}\mathcal{D}_{i}+2x_{i}\alpha_{ii}\mathcal{D}%
_{i}\right) \\
&  +\sum_{i\neq j}\left(  x_{i}x_{j}\mathcal{D}_{i}\mathcal{D}_{j}%
+2x_{i}\alpha_{ij}\mathcal{D}_{j}\right)  .
\end{align*}
Thus%
\begin{gather*}
2\sum\limits_{1\leq i<j\leq N}J_{\varepsilon_{i},\varepsilon_{j}}%
^{2}+2\left\langle x,\nabla_{\kappa}\right\rangle ^{2}=2\left\Vert
x\right\Vert ^{2}\Delta_{\kappa}+\left(  -2\left(  N-1\right)  +2\right)
\left\langle x,\nabla_{\kappa}\right\rangle \\
+2\sum_{i\neq j}\left(  x_{i}\alpha_{ij}\mathcal{D}_{j}+x_{j}\alpha
_{ij}\mathcal{D}_{i}-x_{i}\alpha_{jj}\mathcal{D}_{i}-x_{j}\alpha
_{ii}\mathcal{D}_{j}+2x_{i}\alpha_{ij}\mathcal{D}_{j}\right)  +4\sum_{i=1}%
^{N}x_{i}\alpha_{ii}\mathcal{D}_{i}.
\end{gather*}
Rewrite the second line (with some interchange of indices)
\[
\sum_{i=1}^{N}\left\{  4x_{i}\alpha_{ii}+\sum_{j\neq i}\left(  8x_{j}%
\alpha_{ij}-4x_{i}\alpha_{jj}\right)  \right\}  \mathcal{D}_{i}%
\]
Consider the coefficient of $\mathcal{D}_{i}$ in $\sum\limits_{v\in R_{+}%
}\kappa_{v}\sigma_{v}\left\langle x,\nabla_{\kappa}\right\rangle $ namely
(recall $\left(  \sigma_{v}x\right)  _{i}=x_{i}-2\frac{\left\langle
x,v\right\rangle }{\left\Vert v\right\Vert ^{2}}v_{i}$ and $x_{i}$ is a
multiplication operator)%
\begin{align*}
\sum\limits_{v\in R_{+}}\kappa_{v}\sigma_{v}x_{i}  &  =\sum\limits_{v\in
R_{+}}\kappa_{v}\left\{  x_{i}-\frac{2}{\left\vert v\right\vert ^{2}}%
\sum_{j=1}^{N}x_{j}v_{j}v_{i}\right\}  \sigma_{v}\\
&  =\sum\limits_{v\in R_{+}}\kappa_{v}x_{i}\sigma_{v}-2\sum_{j=1}^{N}%
x_{j}\alpha_{ij}=\sum_{j\neq i}\left(  x_{i}\alpha_{jj}-2x_{j}\alpha
_{ij}\right)  -x_{i}\alpha_{ii}%
\end{align*}
because $\sum\limits_{v\in R_{+}}\kappa_{v}x_{i}\sigma_{v}=\sum_{j=1}^{N}%
x_{i}\alpha_{jj}$ . Thus $2\sum\limits_{1\leq i<j\leq N}J_{\varepsilon
_{i},\varepsilon_{j}}^{2}=-2\left\langle x,\nabla_{\kappa}\right\rangle
^{2}+2\left\vert x\right\vert ^{2}\Delta_{\kappa}-2\left(  N-2\right)
\left\langle x,\nabla_{\kappa}\right\rangle -4\sum\limits_{v\in R_{+}}%
\kappa_{v}\sigma_{v}\left\langle x,\nabla_{\kappa}\right\rangle $.
\end{proof}

\begin{corollary}
$\mathcal{J}$ commutes with the $W\left(  R\right)  $-action.
\end{corollary}

\begin{proof}
$\left\langle x,\nabla_{\kappa}\right\rangle =\delta+\sum\limits_{v\in R_{+}%
}\kappa_{v}\left(  1-\sigma_{v}\right)  $, and $\sum\limits_{v\in R_{+}}%
\kappa_{v}\sigma_{v}$ is in the center of the group algebra because it is a
weighted sum of conjugacy classes.
\end{proof}

\subsection{\label{DHO}The Dunkl harmonic oscillator}

This is the modified Schr\"{o}dinger equation (with parameter $\omega>0$ and
energy $E$)%
\[
\mathcal{H}\psi:=\left(  \omega^{2}\left\vert x\right\vert ^{2}-\Delta
_{\kappa}\right)  \psi=E\psi.
\]
The exponential ground state is $g\left(  x\right)  :=\exp\left(
-\frac{\omega}{2}\left\vert x\right\vert ^{2}\right)  $, as can be seen from
the transformation
\begin{equation}
g^{-1}\left(  \omega^{2}\left\vert x\right\vert ^{2}-\Delta_{\kappa}\right)
\left(  fg\right)  =-\Delta_{\kappa}f+\omega\left(  N+2\gamma_{\kappa}%
+2\delta\right)  f.\label{polHam}%
\end{equation}
which implies $\left(  \omega^{2}\left\vert x\right\vert ^{2}-\Delta_{\kappa
}\right)  g=\omega\left(  N+2\gamma_{\kappa}\right)  g$. Thus the
wavefunctions $\psi$ are of the form $p\left(  x\right)  g\left(  x\right)  $.

In order to transform $\mathcal{H}$ to the usual Schr\"{o}dinger equation with
reflections we introduce the weight function $h_{\kappa}\left(  x\right)
:=\prod\limits_{v\in R_{+}}\left\vert \left\langle x,v\right\rangle
\right\vert ^{\kappa_{v}}$ (a $W\left(  R\right)  $-invariant function,
positively homogeneous of degree $\gamma_{\kappa}$) and find
\begin{equation}
h_{\kappa}\left(  -\Delta_{\kappa}+\omega^{2}\left\vert x\right\vert
^{2}\right)  h_{\kappa}^{-1}=-\Delta+\omega^{2}\left\vert x\right\vert
^{2}+\sum_{v\in R_{+}}\frac{\kappa_{v}\left(  \kappa_{v}-\sigma_{v}\right)
\left\vert v\right\vert ^{2}}{\left\langle x,v\right\rangle ^{2}%
},\label{SchrV}%
\end{equation}
a Schr\"{o}dinger equation with the potential%
\[
V(x)=\omega^{2}\left\vert x\right\vert ^{2}+\sum_{v\in R_{+}}\frac{\kappa
_{v}\left(  \kappa_{v}-\sigma_{v}\right)  \left\vert v\right\vert ^{2}%
}{\left\langle x,v\right\rangle ^{2}},
\]
which includes reflections (for the proof see \cite[app.]{D2022}). The ground
state is $h_{\kappa}g$. For invariant functions the second term becomes
$\sum\limits_{v\in R_{+}}\dfrac{\kappa_{v}\left(  \kappa_{v}-1\right)
\left\vert v\right\vert ^{2}}{\left\langle x,v\right\rangle ^{2}}$. For the
special case where $R$ is the root system of type $A_{N-1}$ and $W\left(
R\right)  =\mathcal{S}_{N}$ (the symmetric group) this potential occurs in the
Calogero-Moser model of $N$ identical particles on a line with $r^{-2}$
interaction potential and harmonic confinement (see Lapointe and Vinet
\cite{LV1996}, also \cite[Sect. 11.6.3]{DX2014}).

Denote $\widetilde{\mathcal{H}}:=-\Delta_{\kappa}+\omega\left(  N+2\gamma
_{\kappa}+2\delta\right)  ,$ an operator on $\mathcal{P}$. In the ordinary
one-variable case the wavefunctions are $H_{n}\left(  \omega^{1/2}s\right)
\exp\left(  -\frac{\omega}{2}s^{2}\right)  $ with eigenvalue $E=\omega\left(
2n+1\right)  $ (Schr\"{o}dinger \cite[pp. 514-516]{S1926}, $H_{n}$ is the
Hermite polynomial of degree $n$).

\begin{definition}
Let $\mathcal{P}_{\leq n}$ denote the space of polynomials of degree $\leq n$
(that is, $\mathcal{P}_{\leq n}=\mathcal{P}_{n}+\mathcal{P}_{n-1}%
+\ldots+\mathcal{P}_{0}$) and let $\pi_{k}$ denote the projection on
$\mathcal{P}_{k}$ for $k\geq0$.
\end{definition}

\begin{proposition}
\label{expwp}Suppose $pg\in$ is a wavefunction and $p\in\mathcal{P}_{\leq n}$,
$\pi_{n}p\neq0$, then $\mathcal{H}\left(  pg\right)  =\omega\left(
N+2\gamma_{\kappa}+n\right)  pg$, $\pi_{n-2i+1}p=0$ for $i\leq\frac{n+1}{2}$,
$\pi_{n-2i}p=\dfrac{1}{i!}\left(  -\dfrac{\Delta_{\kappa}}{4\omega}\right)
^{i}\pi_{n}p$ for $i\leq\frac{n}{2}$ and $p=\exp\left(  -\dfrac{\Delta
_{\kappa}}{4\omega}\right)  \pi_{n}p$.
\end{proposition}

\begin{proof}
Expand the equation $-\Delta_{\kappa}p=\omega\left(  N+2\gamma_{\kappa
}+2\delta\right)  p-Ep$ in homogeneous components%
\[
-\Delta_{\kappa}\pi_{n-i+2}p=\left\{  \omega\left(  N+2\gamma_{\kappa
}+2n-2i\right)  -E\right\}  \pi_{n-i}p,i=0,1,\ldots,n.
\]
The equation at ~$i=1$ implies $E=\omega\left(  N+2\gamma_{\kappa}+2n\right)
$, and thus
\[
-\Delta_{\kappa}\pi_{n-i+2}p=-2i\omega\pi_{n-i}p,i=0,1,\ldots,n.
\]
Since $\pi_{k}p=0$ for $k>n$ it follows that $\pi_{n-1}p=0$ and by induction
that $\pi_{n-1-2i}p=0$ for $i\leq\frac{n-1}{2}$. Use the relation
$-\Delta_{\kappa}\pi_{n-2i+2}p=-4i\omega\pi_{n-2i}p$ inductively to show
$\pi_{n-2i}p=\frac{1}{i!}\left(  -\dfrac{\Delta_{\kappa}}{4\omega}\right)
^{i}\pi_{n}p$ and $\sum_{j=0}^{\left\lfloor n/2\right\rfloor }\dfrac{1}%
{i!}\left(  -\dfrac{\Delta_{\kappa}}{4\omega}\right)  ^{i}\pi_{n}p=\sum
_{j=0}^{\left\lfloor n/2\right\rfloor }\pi_{n-2i}p=p$. The series $\exp\left(
-\dfrac{\Delta_{\kappa}}{4\omega}\right)  \pi_{n}p$ terminates.
\end{proof}

\begin{corollary}
\label{exp2wav}With the same hypotheses $\exp\left(  \dfrac{\Delta_{\kappa}%
}{4\omega}\right)  p=\pi_{n}p$.
\end{corollary}

For $n=0,1,2,\ldots$ set $E_{n}:=\omega\left(  N+2\gamma_{\kappa}+2n\right)
$. It follows that $\mathcal{H}\left(  pg\right)  =E_{n}pg$ implies
$p\in\mathcal{P}_{\leq n}$ and $\pi_{n}p\neq0$. Also $\phi\in\mathcal{P}%
_{n,\kappa}$ implies $\mathcal{H}\left(  \phi g\right)  =E_{n}\phi g$.

The Laguerre polynomial of degree $m$ and index $\alpha>-1$ satisfies%
\begin{align}
L_{n}^{\left(  \alpha\right)  }\left(  s\right)   &  =\frac{\left(
\alpha+1\right)  _{n}}{n!}\sum_{j=0}^{n}\frac{\left(  -n\right)  _{j}}{\left(
\alpha+1\right)  _{j}}\frac{s^{j}}{j!}\label{Lagdef}\\
\int_{0}^{\infty}L_{m}^{\left(  \alpha\right)  }\left(  s\right)
L_{k}^{\left(  \alpha\right)  }\left(  s\right)  s^{\alpha}e^{-s}\mathrm{d}s
&  =\delta_{mk}\frac{\Gamma\left(  \alpha+1+m\right)  }{m!}=\delta_{mk}%
\Gamma\left(  \alpha+1\right)  \frac{\left(  \alpha+1\right)  _{m}}{m!}.
\label{orthoLag}%
\end{align}

\begin{proposition}
\label{hmXLag}Suppose $p\in\mathcal{P}_{n,\kappa}$ and $m,n=1,2,3,\ldots$then
$\psi_{n,m}\left(  x\right)  :=L_{m}^{\left(  \alpha\right)  }\left(
\omega\left\vert x\right\vert ^{2}\right)  p\left(  x\right)  g\left(
x\right)  $ is a wavefunction with energy $E_{n+2m},$where $\alpha=\frac{N}%
{2}+\gamma_{\kappa}+n-1$.
\end{proposition}

\begin{proof}
By $($\ref{Lagdef}) $\pi_{n+2m}f_{n,m}=\frac{\left(  -1\right)  ^{m}}%
{m!}\omega^{m}\left\vert x\right\vert ^{2m}p\left(  x\right)  $ and%
\begin{align*}
\exp\left(  -\frac{\Delta_{\kappa}}{4\omega}\right)  \pi_{n+2m}f_{n,m}  &
=\frac{1}{m!}\sum_{i=0}^{m}\frac{\left(  -1\right)  ^{m-i}}{4^{i}i!}%
\omega^{m-i}\Delta_{\kappa}^{i}\left\vert x\right\vert ^{2m}p\left(  x\right)
\\
&  =\frac{1}{m!}\sum_{i=0}^{m}\frac{\left(  -1\right)  ^{m-i}}{4^{i}i!}%
\omega^{m-i}\prod\limits_{j=0}^{i-1}C_{n,m-j+1}~\left\vert x\right\vert
^{2m-2i}p\left(  x\right)
\end{align*}
(recall $C_{n,k}=4k\left(  \frac{N}{2}+\gamma_{\kappa}+n+k-1\right)  $ and
$C_{n,0}=1$ from Proposition \ref{Lap^m}, and $\prod\limits_{j=0}%
^{i-1}C_{n,m-i+1}=2^{2i}\frac{m!}{\left(  m-i\right)  !}\left(  \frac{N}%
{2}+\gamma_{\kappa}+n+m-i\right)  _{i}$ so the multiplier of $p$ is
\begin{align*}
&  \sum_{i=0}^{m}\dfrac{\left(  -1\right)  ^{m-i}}{i!\left(  m-i\right)
!}\left(  \frac{N}{2}+\gamma_{\kappa}+n-i\right)  _{i}\left(  \omega\left\vert
x\right\vert ^{2}\right)  ^{m-i}\\
&  =\frac{1}{m!}\sum_{j=0}^{m}\frac{\left(  -m\right)  _{j}\left(  \frac{N}%
{2}+\gamma_{\kappa}+n\right)  _{m}}{j!\left(  \frac{N}{2}+\gamma_{\kappa
}+n\right)  _{j}}\left(  \omega\left\vert x\right\vert ^{2}\right)  ^{j}%
=L_{m}^{\left(  N/2+\gamma_{\kappa}+n-1\right)  }\left(  \omega\left\vert
x\right\vert ^{2}\right)  ,
\end{align*}
changing the index of summation $i=m-j$.
\end{proof}

\begin{corollary}
\label{delLagh}With the above hypotheses%
\[
\Delta_{\kappa}L_{m}^{\left(  \alpha\right)  }\left(  \omega\left\vert
x\right\vert ^{2}\right)  p\left(  x\right)  =-4\omega\left(  \alpha+m\right)
L_{m-1}^{\left(  \alpha\right)  }\left(  \omega\left\vert x\right\vert
^{2}\right)  p\left(  x\right)  .
\]

\end{corollary}

\begin{proof}
Apply formula (\ref{delxsq1}) to $\pi_{n+2m}f_{n,m}$ which results in
$-4\omega\left(  \alpha+m\right)  \pi_{n+2m-2}f_{n,m-1}$; and $\Delta_{\kappa
}$ commutes with $\exp\left(  -\dfrac{\Delta_{\kappa}}{4\omega}\right)  $.
\end{proof}

\begin{proposition}
\label{wav2laghm}Suppose $pg$ is a wavefunction with $\widetilde{\mathcal{H}%
}p=E_{n}p$, $\pi_{n}p\neq0$ then%
\[
p\left(  x\right)  =\sum_{j=0}^{\left\lfloor n/2\right\rfloor }\dfrac{\left(
-1\right)  ^{j}}{\left(  4\omega\right)  ^{j}\left(  N/2+\gamma_{k}%
+n-2j\right)  _{j}}L_{j}^{\left(  N/2+\gamma_{\kappa}+n-2j-1\right)  }\left(
\omega\left\vert x\right\vert ^{2}\right)  \Lambda_{n-2j}\Delta_{\kappa}%
^{j}\pi_{n}p\left(  x\right)  .
\]

\end{proposition}

\begin{proof}
By Proposition \ref{hmXLag} if $q\in\mathcal{P}_{n-2j,\kappa}$ then%
\[
\exp\left(  -\frac{\Delta_{\kappa}}{4\omega}\right)  \left\vert x\right\vert
^{2j}q\left(  x\right)  =\frac{\left(  -1\right)  ^{j}j!}{\omega^{j}}%
L_{j}^{\left(  N/2+\gamma_{\kappa}+n-2j-1\right)  }\left(  \omega\left\vert
x\right\vert ^{2}\right)  q\left(  x\right)  ,
\]
apply this formula term-by-term to the identity in Proposition \ref{hmex} and
use $\exp\left(  -\frac{\Delta_{\kappa}}{4\omega}\right)  \pi_{n}p=p$ from
Corollary \ref{exp2wav}.
\end{proof}

The Proposition shows that any wavefunction can be expressed as a sum of
products of (radial) Laguerre polynomials with homogeneous harmonic
polynomials (and with $g$).

We introduce \textit{raising} and \textit{lowering} operators ($\left\{
A,B\right\}  :=AB+BA$):

\begin{definition}
For $a\in\mathbb{R}^{N},a\neq0$ and $\varepsilon=\pm1$ let $A_{a}%
^{\varepsilon}=\omega\left\langle a,x\right\rangle +\varepsilon\left\langle
a,\nabla_{\kappa}\right\rangle $ and $H_{a}:=\frac{1}{2}\left\{  A_{a}%
^{+},A_{a}^{-}\right\}  =\omega^{2}\left\langle a,x\right\rangle
^{2}-\left\langle a,\nabla_{\kappa}\right\rangle ^{2}$.
\end{definition}

\begin{proposition}
\label{[hH'0} $g^{-1}A_{a}^{+}g=\left\langle a,\nabla_{\kappa}\right\rangle $
(\textit{lowering}) and $g^{-1}A_{a}^{-}g=2\omega\left\langle a,x\right\rangle
-\left\langle a,\nabla_{\kappa}\right\rangle $ (\textit{raising}); and
$\left[  \mathcal{H},H_{a}\right]  =0$. Also $g^{-1}H_{a}g=\omega\left(
\left\langle a,x\right\rangle \left\langle a,\nabla_{\kappa}\right\rangle
+\left\langle a,\nabla_{\kappa}\right\rangle \left\langle a,x\right\rangle
\right)  -\left\langle a,\nabla_{\kappa}\right\rangle ^{2}$.
\end{proposition}

\begin{proof}
The commutator
\[
\left[  \omega^{2}\left\langle a,x\right\rangle ^{2}-\left\langle
a,\nabla_{\kappa}\right\rangle ^{2},\ \mathcal{H}\right]  =-\left[
\left\langle a,\nabla_{\kappa}\right\rangle ^{2},\omega^{2}\left\vert
x\right\vert ^{2}\right]  -\left[  \omega^{2}\left\langle a,x\right\rangle
^{2},\Delta_{\kappa}\right]
\]
and expanding the right hand side with (\ref{Dbxsq}) and $\left[
A^{2},B\right]  =A\left[  A,B\right]  +\left[  A,B\right]  A$ shows $\left[
H_{a},\mathcal{H}\right]  =0$.
\end{proof}

\begin{proposition}
\label{wHaw}Suppose $w\in W\left(  R\right)  $ then $w^{-1}H_{a}w=H_{aw}$;
suppose $S\subset\mathbb{R}^{N}$ and $S\cup\left(  -S\right)  $ is an
$W\left(  R\right)  $-orbit (closed under $v\rightarrow vw$) then $\sum_{v\in
S}H_{v}^{k}$ commutes with each $w\in W\left(  R\right)  $, for
$k=1,2,3,\ldots.$
\end{proposition}

\begin{proof}
This follows from $\left\langle a,\nabla_{\kappa}\right\rangle w=w\left\langle
aw,\nabla_{\kappa}\right\rangle $ (see\cite[Prop. 6.4.3]{DX2014}) and
$w\left(  \left\langle aw,x\right\rangle p\left(  x\right)  \right)
=\left\langle aw,xw\right\rangle p\left(  xw\right)  =\left\langle
a,x\right\rangle wp\left(  x\right)  $ (because $w\in O_{N}\left(
\mathbb{R}\right)  $).
\end{proof}

From Proposition \ref{Jprops} it follows that $\left[  \widetilde{\mathcal{H}%
},J_{a,b}\right]  =0$, and thus $\left[  \widetilde{\mathcal{H}}%
,\mathcal{J}\right]  =0$.

We now have a collection of operators commuting with $W\left(  R\right)  $ and
$\mathcal{H}$. See Quesne \cite{Q2010} who used these operators to prove
superintegrability for the Dunkl wave equation in the even dihedral group
system. Tremblay, Turbiner and Winternitz \cite{TTW2009} also investigated the
dihedral potentials. Genest and Vinet \cite{GV2014} gave a detailed analysis
of the various bases of wavefunctions for the $\mathbb{Z}_{2}^{3}$-system
(abelian group) in $\mathbb{R}^{3}$, see also G., V. and Zhedanov
\cite{GVZ2014}.

\begin{proposition}
Suppose $\varepsilon=\pm1$ then $\left[  \mathcal{H},A_{a}^{\varepsilon
}\right]  =-2\omega\varepsilon A_{a}^{\varepsilon}$
\end{proposition}

\begin{proof}
$\left[  \left\langle a,\nabla_{\kappa}\right\rangle ,\left\vert x\right\vert
^{2}\right]  =2\left\langle a,x\right\rangle $ and $\left[  \Delta_{\kappa
},\left\langle a,x\right\rangle \right]  =2\left\langle a,\nabla_{\kappa
}\right\rangle $ (from (\ref{Dbxsq} )thus%
\begin{align*}
\left[  \omega\left\langle a,x\right\rangle +\varepsilon\left\langle
a,\nabla_{\kappa}\right\rangle ,\left\vert x\right\vert ^{2}\right]   &
=2\varepsilon\left\langle a,x\right\rangle \\
\left[  \omega\left\langle a,x\right\rangle +\varepsilon\left\langle
a,\nabla_{\kappa}\right\rangle ,\Delta_{\kappa}\right]   &  =-2\omega
\left\langle a,\nabla_{\kappa}\right\rangle
\end{align*}
and%
\begin{align*}
\left[  \omega^{2}\left\vert x\right\vert ^{2}-\Delta_{\kappa},\omega
\left\langle a,x\right\rangle +\varepsilon\left\langle a,\nabla_{\kappa
}\right\rangle \right]   &  =-2\varepsilon\omega^{2}\left\langle
a,x\right\rangle -2\omega\left\langle a,\nabla_{\kappa}\right\rangle \\
&  =-2\omega\varepsilon\left(  \omega\left\langle a,x\right\rangle
+\varepsilon\left\langle a,\nabla_{\kappa}\right\rangle \right)
\end{align*}
since $\varepsilon^{2}=1$.
\end{proof}

\begin{corollary}
Suppose $\mathcal{H}f=\lambda f$ then $\mathcal{H}A_{a}^{\varepsilon}f=\left(
\lambda-2\omega\varepsilon\right)  A_{a}^{\varepsilon}f$
\end{corollary}

\begin{proof}
$\mathcal{HA}_{a}^{\varepsilon}f=A_{a}^{\varepsilon}\mathcal{H}-2\omega
\varepsilon A_{a}^{\varepsilon}=A_{a}^{\varepsilon}\left(  \mathcal{H}%
-2\omega\varepsilon\right)  .$
\end{proof}

When $\mathcal{H}f=\lambda f,~f=pg$ and $p\in\mathcal{P}_{\leq n},\pi_{n}%
p\neq0$ then $\lambda=E_{n}$. Hence $A_{a}^{-}f=\left(  \left(  2\omega
\left\langle a,x\right\rangle -\left\langle a,\nabla_{\kappa}\right\rangle
\right)  p\right)  g$ is an $\mathcal{H}$-eigenfunction for $E_{n}%
+2\omega=E_{n+1}$ and $A_{a}^{+}f=\left(  \left\langle a,\nabla_{\kappa
}\right\rangle p\right)  g$ is an $\mathcal{H}$-eigenfunction for
$E_{n}-2\omega=E_{n-1}$.

\subsection{Integrals and inner products}

The quantum-mechanical interpretation of $\left\vert \psi\right\vert ^{2}$ as
a probability distribution when $\psi$ is a wavefunction is in terms of the
Schr\"{o}dinger equation (\ref{SchrV}), that is, the formulation $\psi\left(
x\right)  =p\left(  x\right)  g\left(  x\right)  h_{\kappa}\left(  x\right)  $
where $p$ is a polynomial eigenfunction of (\ref{polHam}),. Accordingly we
introduce the Hilbert space $L^{2}\left(  \mathbb{R}^{N},h_{\kappa}%
^{2}\mathrm{d}m\right)  $ where $\mathrm{d}m$ denotes the Lebesgue measure.

\begin{theorem}
(\cite[7.7.9]{DX2014}) Suppose $p,q$ are sufficiently smooth and $pq$ has
exponential decay then (for $1\leq i\leq N$)%
\begin{equation}
\int_{\mathbb{R}^{N}}\left(  \mathcal{D}_{i}p\right)  qh_{\kappa}%
^{2}\mathrm{d}m=-\int_{\mathbb{R}^{N}}p\left(  \mathcal{D}_{i}q\right)
h_{\kappa}^{2}\mathrm{d}m,\label{D*D}%
\end{equation}
Thus the adjoint of $\mathcal{D}_{i}$ is defined on a dense subspace of
$L^{2}\left(  \mathbb{R}^{N},h_{\kappa}^{2}\mathrm{d}m\right)  $ and
$\mathcal{D}_{i}^{\ast}=-\mathcal{D}_{i}$ .
\end{theorem}

We will use this meaning of adjoint throughout.

\begin{corollary}
(1) $\mathcal{H}^{\ast}=\mathcal{H}$ : (2) $J_{a,b}^{\ast}=-J_{a,b}~$; (3)
$\mathcal{J}^{\ast}=\mathcal{J}$; (4) $\left(  A_{a}^{+}\right)  ^{\ast}%
=A_{a}^{-}~$; (5) $H_{a}^{\ast}=H_{a}$.
\end{corollary}

There is a normalization constant%
\begin{equation}
c_{\kappa,\omega}^{-1}:=\int_{\mathbb{R}^{N}}\exp\left(  -\omega\left\vert
x\right\vert ^{2}\right)  h_{\kappa}\left(  x\right)  ^{2}\mathrm{d}m\left(
x\right)  .\label{MacdSel}%
\end{equation}
For the root systems $A_{N-1}$ and $B_{N}$ this is called the
Macdonald-Mehta-Selberg integral. Etingof \cite{E2010}  gave a unified proof
of the evaluation for any reflection group with only one conjugacy class of
reflections. We will use two inner products for polynomials, one is of
integral type and the other is algebraic.

\begin{definition}
For $p,q\in\mathcal{P}$ let%
\begin{align*}
\left\langle p,q\right\rangle _{2}  &  =c_{\kappa,\omega}\int_{\mathbb{R}^{N}%
}p\left(  x\right)  q\left(  x\right)  \exp\left(  -\omega\left\vert
x\right\vert ^{2}\right)  h_{\kappa}\left(  x\right)  ^{2}\mathrm{d}m\left(
x\right)  ,\\
\left\langle p,q\right\rangle _{\kappa,\omega}  &  =p\left(  \frac{1}{2\omega
}\mathcal{D}_{1}.\ldots,\frac{1}{2\omega}\mathcal{D}_{N}\right)  q\left(
x\right)  |_{x=0}.
\end{align*}

\end{definition}

These bilinear forms satisfy $\left\langle pw,qw\right\rangle _{2}%
=\left\langle p,q\right\rangle _{2}$, $\left\langle pw,qw\right\rangle
_{\kappa,\omega}=\left\langle p,q\right\rangle _{\kappa,\omega}=\left\langle
q,p\right\rangle _{\kappa,\omega}$ for $w\in W\left(  R\right)  $ (see
\cite[Thm 7.2.3]{DX2014}). Also $\left\langle \mathcal{D}_{i}p,q\right\rangle
_{\kappa,\omega}=\left\langle p,2\omega x_{i}q\right\rangle _{\kappa,\omega}$,
and $p\in\mathcal{P}_{n},q\in\mathcal{P}_{k}$ and $n\neq k$ implies
$\left\langle p,q\right\rangle _{\kappa,\omega}=0$. The relation between the
two inner products follows from a lemma. Set $\boldsymbol{E:}=\exp\left(
\dfrac{\Delta_{\kappa}}{4\omega}\right)  $ (locally finite on $\mathcal{P}$).

\begin{lemma}
$\left[  x_{i},\Delta_{\kappa}^{n}\right]  =-2n\Delta_{\kappa}^{n-1}%
\mathcal{D}_{i}$ and $\left[  x_{i},\boldsymbol{E}\right]  =-\dfrac{1}%
{2\omega}\boldsymbol{E}\mathcal{D}_{i}$, for $1\leq i\leq N$. Also
$\left\langle \boldsymbol{E}\mathcal{D}_{i}p,\boldsymbol{E}q\right\rangle
_{\kappa,\omega}=\allowbreak\left\langle \boldsymbol{E}p,\boldsymbol{E}\left(
2\omega x_{i}-\mathcal{D}_{i}\right)  q\right\rangle _{\kappa,\omega}$.
\end{lemma}

\begin{proof}
From $\left[  \Delta_{\kappa},\left\langle a,x\right\rangle \right]
=2\left\langle a,\nabla_{\kappa}\right\rangle $ it follows that $\left[
x_{i},\Delta_{\kappa}\right]  =-2\mathcal{D}_{i}$. Use $\left[  A,B_{1}%
B_{2}\right]  =\left[  A,B_{1}\right]  B_{2}+B_{1}\left[  A,B_{2}\right]  $
and induction to obtain%
\begin{align*}
\left[  x_{i},\Delta_{\kappa}^{n}\right]   &  =\left[  x_{i},\Delta_{\kappa
}^{n-1}\right]  \Delta_{\kappa}+\Delta_{\kappa}^{n-1}\left[  x_{i}%
,\Delta_{\varphi}\right] \\
&  =-2\left(  n-1\right)  \Delta_{\kappa}^{n-2}\mathcal{D}_{i}\Delta_{\kappa
}+\Delta_{\kappa}^{n-1}\left(  -2\mathcal{D}_{i}\right)  =-2n\Delta_{\kappa
}^{n-1}\mathcal{D}_{i}.
\end{align*}
Thus%
\[
\left[  x_{i},\boldsymbol{E}\right]  =\sum_{n=0}^{\infty}\frac{1}{\left(
4\omega\right)  ^{n}n!}\left[  x_{i},\Delta_{\kappa}^{n}\right]  =-2\sum
_{n=1}^{\infty}\frac{n}{\left(  4\omega\right)  ^{n}n!}\Delta_{\kappa}%
^{n-1}\mathcal{D}_{i}=-\frac{2}{4\omega}\boldsymbol{E}\mathcal{D}_{i}.
\]
For the second part%
\[
\left\langle \boldsymbol{E}\mathcal{D}_{i}p,\boldsymbol{E}q\right\rangle
_{\kappa,\omega}=\left\langle \mathcal{D}_{i}\boldsymbol{E}p,\boldsymbol{E}%
q\right\rangle _{\kappa,\omega}=\left\langle \boldsymbol{E}p,2\omega
x_{i}\boldsymbol{E}q\right\rangle _{\kappa,\omega}=\left\langle \boldsymbol{E}%
p,\boldsymbol{E}\left(  2\omega x_{i}-\mathcal{D}_{i}\right)  q\right\rangle
_{\kappa,\omega}.
\]

\end{proof}

\begin{theorem}
Suppose $p,q\in\mathcal{P}$ then $\left\langle \boldsymbol{E}p,\boldsymbol{E}%
q\right\rangle _{\kappa,\omega}=\left\langle p,q\right\rangle _{2}$.
\end{theorem}

\begin{proof}
Temporarily denote $\left\langle p,q\right\rangle _{\boldsymbol{E}%
}=\left\langle \boldsymbol{E}p,\boldsymbol{E}q\right\rangle _{\kappa,\omega}$.
We will show that $\left\langle p,q\right\rangle _{\boldsymbol{E}}$ and
$\left\langle p,q\right\rangle _{2}$ satisfy the same recurrence relations and
$\left\langle 1,1\right\rangle _{\boldsymbol{E}}=1=\left\langle
1,1\right\rangle _{2}$. From%
\[
\int_{\mathbb{R}^{N}}\left(  \mathcal{D}_{i}p\right)  \left(  qe^{-\omega
\left\vert x\right\vert ^{2}}\right)  h_{\kappa}^{2}\mathrm{d}m=-\int
_{\mathbb{R}^{N}}p\mathcal{D}_{i}\left(  qe^{-\omega\left\vert x\right\vert
^{2}}\right)  h_{\kappa}^{2}\mathrm{d}m
\]
it follows that $\left\langle \mathcal{D}_{i}p,q\right\rangle _{2}%
=-\left\langle p,-2\omega x_{i}q+\mathcal{D}_{i}q\right\rangle _{2}$. From the
Lemma $2\omega\left\langle p,x_{i}q\right\rangle _{\boldsymbol{E}%
}=\left\langle \mathcal{D}_{i}p,q\right\rangle _{\boldsymbol{E}}+\left\langle
p,\mathcal{D}_{i}q\right\rangle _{\boldsymbol{E}}$ and by symmetry
$\left\langle p,x_{i}q\right\rangle _{\boldsymbol{E}}=\left\langle
x_{i}p,q\right\rangle _{\boldsymbol{E}}$. This together with $\left\langle
1,1\right\rangle _{\boldsymbol{E}}=1$ shows $\left\langle pq,1\right\rangle
_{\boldsymbol{E}}=\left\langle p,q\right\rangle _{\boldsymbol{E}}$ (by
definition $\left\langle pq,1\right\rangle _{2}=\left\langle p.q\right\rangle
_{2}$). Also $\left\langle 1,\left(  2\omega x_{i}-\mathcal{D}_{i}\right)
p\right\rangle _{\boldsymbol{E}}=\left\langle \mathcal{D}_{i}1,p\right\rangle
_{\boldsymbol{E}}=0$. Inductively suppose $\left\langle p,1\right\rangle
_{\boldsymbol{E}}=\left\langle p,1\right\rangle _{2}$ for all $p\in
\mathcal{P}_{\leq n}$; any $p\in\mathcal{P}_{\leq n+1}$ can be expressed as
$p\left(  x\right)  =\sum_{i=1}^{N}x_{i}p_{i}\left(  x\right)  $ with each
$p_{i}\in\mathcal{P}_{\leq n}$, then $\left\langle p,1\right\rangle
_{\boldsymbol{E}}=\sum_{i=1}^{N}\left\langle x_{i}p_{i},1\right\rangle
_{\boldsymbol{E}}\allowbreak=\frac{1}{2\omega}\sum_{i=1}^{N}\left\langle
\mathcal{D}_{i}p_{i},1\right\rangle _{\boldsymbol{E}}=\frac{1}{2\omega}%
\sum_{i=1}^{N}\left\langle \mathcal{D}_{i}p_{i},1\right\rangle _{2}$. By the
analogous argument $\frac{1}{2\omega}\sum_{i=1}^{N}\left\langle \mathcal{D}%
_{i}p_{i},1\right\rangle _{2}=\left\langle p,1\right\rangle _{2}$. Thus
$\left\langle p,1\right\rangle _{\boldsymbol{E}}=\left\langle p,1\right\rangle
_{2}$ for all $p\in\mathcal{P}$, in particular $\left\langle p,q\right\rangle
_{\boldsymbol{E}}=\left\langle pq,1\right\rangle _{\boldsymbol{E}%
}=\left\langle p,q\right\rangle _{2}.$
\end{proof}

This result has a striking consequence for wavefunctions.

\begin{theorem}
\label{wavenorm}Suppose $p,q\in$ $\mathcal{P}$ and $\widetilde{\mathcal{H}%
}p=E_{n}p,\widetilde{\mathcal{H}}q=E_{k}q$. If $n=k$ then $\left\langle
p,q\right\rangle _{2}=\left\langle \pi_{n}p,\pi_{n}q\right\rangle
_{\kappa,\omega}$. If $n\neq k$ then $\left\langle p,q\right\rangle _{2}=0$.
\end{theorem}

\begin{proof}
The second part follows from $\widetilde{\mathcal{H}}^{\ast}=\widetilde
{\mathcal{H}}$. Suppose $k=n$ then by Corollary \ref{exp2wav} $\boldsymbol{E}%
p=\pi_{n}p$, $\boldsymbol{E}q=\pi_{n}q$ and $\left\langle p,q\right\rangle
_{2}=\left\langle \boldsymbol{E}p,\boldsymbol{E}q\right\rangle _{\kappa
,\omega}=\left\langle \pi_{n}p,\pi_{n}q\right\rangle _{\kappa,\omega}$.
\end{proof}

The regrettable fact is that it is very difficult to produce explicit harmonic
polynomials, orthogonal bases, or $L^{2}\left(  \mathbb{R}^{N},h_{\kappa}%
^{2}\mathrm{d}m\right)  $-norms for specific root systems of rank greater than
2. With symbolic computation one can find polynomials of low degree and their
norms, but these methods do not produce explicit bases for all degrees. There
is one modest result on norms, extending Proposition \ref{hmXLag}.

\begin{proposition}
\label{Laghm}Suppose $p,q\in\mathcal{P}_{n,\kappa}$ and $m,n=1,2,3,\ldots$and
$p_{n,m}\left(  x\right)  :=L_{m}^{\left(  \alpha\right)  }\left(
\omega\left\vert x\right\vert ^{2}\right)  p\left(  x\right)  $,
$q_{n,m}\left(  x\right)  :=L_{m}^{\left(  \alpha\right)  }\left(
\omega\left\vert x\right\vert ^{2}\right)  q\left(  x\right)  $ ($\alpha
=\frac{N}{2}+\gamma_{\kappa}+n-1$) then $\left\langle p_{n,m},q_{n,m}%
\right\rangle _{2}=\frac{1}{n!}\left(  \frac{N}{2}+\gamma_{\kappa}+n\right)
_{m}\left\langle p,q\right\rangle _{2}$.
\end{proposition}

\begin{proof}
$\pi_{n+2m}p_{n,m}=\frac{\left(  -1\right)  ^{m}}{m!}\omega^{m}\left\vert
x\right\vert ^{2m}p\left(  x\right)  $ (similarly for $\pi_{n+2m}q_{n,m}$) and
by the Theorem%
\begin{gather*}
\left\langle p_{n,m},q_{n,m}\right\rangle _{2}=\pi_{n+2m}p_{n,m}\left(
\frac{1}{2\omega}\nabla_{\kappa}\right)  \pi_{n+2m}q_{n,m}\left(  x\right)  \\
=\frac{1}{2^{2m}\left(  m!\right)  ^{2}}\Delta_{\kappa}^{m}p\left(  \frac
{1}{2\omega}\nabla_{\kappa}\right)  \left\vert x\right\vert ^{2m}q\left(
x\right)  =\frac{1}{2^{2m}\left(  m!\right)  ^{2}}p\left(  \frac{1}{2\omega
}\nabla_{\kappa}\right)  \Delta_{\kappa}^{m}\left\vert x\right\vert
^{2m}q\left(  x\right)  \\
=\frac{1}{2^{2m}\left(  m!\right)  ^{2}}p\left(  \frac{1}{2\omega}%
\nabla_{\kappa}\right)  2^{2m}m!\left(  \frac{N}{2}+\gamma_{\kappa}+n\right)
_{m}~q\left(  x\right)  \\
=\frac{1}{m!}\left(  \frac{N}{2}+\gamma_{\kappa}+n\right)  _{m}~p\left(
\frac{1}{2\omega}\nabla_{\kappa}\right)  q\left(  x\right)  =\frac{1}%
{m!}\left(  \frac{N}{2}+\gamma_{\kappa}+n\right)  _{m}~\left\langle
p,q\right\rangle _{2}.
\end{gather*}
Note $\boldsymbol{E}p=p$ since $\Delta_{\kappa}p=0$.
\end{proof}

\section{\label{IcosGrp}The icosahedral group}

\subsection{Geometric properties}

The positive root system of type $H_{3}$ is
\[
R_{+}=\left\{
\begin{array}
[c]{c}%
(2,0,0),(0,2,0),(0,0,2),(\tau,\pm\tau^{-1},\pm1),\\
(\pm1,\tau,\pm\tau^{-1}),(\tau^{-1},\pm1,\tau),(-\tau^{-1},1,\tau),(\tau
^{-1},1,-\tau)
\end{array}
\right\}  ,
\]
where the choices of signs in $\pm$ are independent of each other and the
roots satisfy $\left\langle v,u_{0}\right\rangle >0$, $u_{0}:=\left(
3,2\tau,1\right)  $, and $\tau$ denotes the \textit{golden ratio
}\footnote{Some use $\phi$ for the golden ratio, while we follow Coxeter's
\cite{HSC1973} usage.} $\left(  1+\sqrt{5}\right)  /2$. Thus $\tau^{2}=\tau
+1$. Henceforth let $G:=W\left(  H_{3}\right)  $, the icosahedral group. It is
the symmetry group of the regular \textit{icosahedron},%
\[
\mathcal{I}:=\{\left(  0,\pm\tau,\pm1\right)  ,\left(  \pm1,0,\pm\tau\right)
,\left(  \pm\tau,\pm1,0\right)  \}
\]
(12 vertices, 20 triangular faces) and of the regular \textit{dodecahedron}%
\[
\mathcal{K}:=\{\left(  0,\pm\tau^{-1},\pm\tau\right)  ,\allowbreak\left(
\pm\tau,0,\pm\tau^{-1}\right)  ,\left(  \pm\tau^{-1},\pm\tau,0\right)
,\left(  \pm1,\pm1,\pm1\right)  \}
\]
~ ( 20 vertices, 12 pentagonal faces); see Coxeter \cite[Ch. II]{HSC1973} for
the details. Let $\mathcal{I}^{\prime}$, $\mathcal{K}^{\prime}$ denote the
points of $\mathcal{I}$, $\mathcal{K}$ normalized to lie on the unit sphere.
The 15 great circles $\left\{  x:\left\vert x\right\vert =1,\left\langle
x,v\right\rangle =0\right\}  $ for $v\in R_{+}$ form a spherical complex of
120 triangular regions whose vertices are $\mathcal{I}^{\prime}\cup
\mathcal{K}^{\prime}$ together with the midpoints of the 30 edges of
$\mathcal{K}^{\prime}$ (equivalently the vertices, centers of faces and
centers of edges of either $\mathcal{I}^{\prime}$ or $\mathcal{K}^{\prime}$).
From the general theory of reflection groups $\#G=120$. The interior angles of
any of the spherical triangles are $\frac{\pi}{2},\frac{\pi}{3},\frac{\pi}{5}%
$. The reflections along the simple roots $v_{1}=(\tau,-\tau^{-1},-1)$,
$v_{2}=(-1,\tau,-\tau^{-1})$, $v_{3}=(\tau^{-1},-1,\tau)$ generate $G$ and the
region $\left\{  x:\left\langle x,v_{i}\right\rangle >0,i=1,2,3\right\}  $ is
the fundamental region, intersecting the unit sphere at $\frac{1}{\sqrt
{\tau+2}}(\tau,1,0)$, $\frac{1}{\sqrt{3}}(1,1,1)$, $\frac{1}{2}(1,\tau
,\tau^{-1})$. The fundamental degrees are $2,6,10$ and the ring $\mathcal{P}%
^{G}$ of $G$-invariant polynomials is generated by $\left\vert x\right\vert
^{2},\prod\limits_{u\in\mathcal{I}_{+}}\left\langle u,x\right\rangle
,\prod\limits_{u\in\mathcal{K}_{+}}\left\langle u,x\right\rangle $ where
$\mathcal{I}_{+}:=\left\{  u\in\mathcal{I}:\left\langle u,u_{0}\right\rangle
>0\right\}  $(similarly for $\mathcal{K}_{+}$, use one factor from each
antipodal pair of vertices). The products are invariant because the action of
any reflection changes an even number of signs. The polynomial $\left(
1+t\right)  \left(  1+5t\right)  \left(  1+9t\right)  =1+15t+59t^{2}+45t^{3}$
shows that $G$ has 15 reflections, 59 plane rotations, and 45 transformations
with no fixed vector.

\subsection{Analytic aspects}

There is just one conjugacy class of reflections in $G$ and there is one
parameter $\kappa$, so that $\gamma_{\kappa}=15\kappa$. Furthermore
$\dim\mathcal{P}_{n}=\binom{n+2}{2}$, so that the multiplicity of the energy
eigenvalue $E_{n}=\omega\left(  3+30\kappa+2n\right)  $ is $\binom{n+2}{2}$.
From the decomposition $\mathcal{P}_{n}=\mathcal{P}_{n,\kappa}\oplus\left\vert
x\right\vert ^{2}\mathcal{P}_{n-2}$ it follows $\dim\mathcal{P}_{n,\kappa
}=2n+1$. Let $\mathcal{P}^{G}:=\left\{  p\in\mathcal{P}:wp=p~\forall w\in
G\right\}  $, the $G$-invariant polynomials. The Poincar\'{e} series for
$\mathcal{P}^{G}$ and the harmonic invariants are%
\begin{align*}
\sum_{n=0}^{\infty}\dim\left(  \mathcal{P}^{G}\cap\mathcal{P}_{n}\right)
t^{n}  &  =\left[  \left(  1-t^{2}\right)  \left(  1-t^{6}\right)  \left(
1-t^{10}\right)  \right]  ^{-1}\\
\sum_{n=0}^{\infty}\dim\left(  \mathcal{P}^{G}\cap\mathcal{P}_{n,\kappa
}\right)  t^{n}  &  =\left[  \left(  1-t^{6}\right)  \left(  1-t^{10}\right)
\right]  ^{-1}%
\end{align*}
respectively. Note that the lowest degree with $\dim\left(  \mathcal{P}%
^{G}\cap\mathcal{P}_{n,\kappa}\right)  \geq2$ is $n=30$.

For each $y\in\mathcal{I}$ there are 5 reflections that fix $y$ and the other
10 map $y$ onto the points of $\mathcal{I}\backslash\left\{  \pm y\right\}  $.
This fact will be used in the construction of special polynomials for which it
is easy to compute the actions of $\nabla_{\kappa}$ and $\Delta_{\kappa}$.

The normalization constant (\ref{MacdSel}) for $H_{3}$ was first proven by F.
Garvan\cite{G1989} (more than 100 linear equations are constructed with
conceptual methods, then solved by computer); later Etingof \cite{E2010} gave
a proof valid for all reflection groups with one conjugacy class of reflections.

\begin{theorem}
For $\kappa>0$, $\omega>0$ and $h_{\kappa}\left(  x\right)  :=\prod
\limits_{v\in R_{+}}\left\vert \left\langle x,v\right\rangle \right\vert
^{\kappa}$%
\[
c_{\kappa,\omega}^{-1}:=\int_{\mathbb{R}^{3}}h_{\kappa}\left(  x\right)
^{2}\exp\left(  -\omega\left\vert x\right\vert ^{2}\right)  \mathrm{d}%
x=\left(  \frac{\pi}{\omega}\right)  ^{\frac{3}{2}}\frac{\Gamma\left(
2\kappa+1\right)  \Gamma\left(  6\kappa+1\right)  \Gamma\left(  10\kappa
+1\right)  }{\omega^{15\kappa}\Gamma\left(  \kappa+1\right)  ^{3}}.
\]

\end{theorem}

\begin{proof}
The general formula from Etingof uses $\omega=\frac{1}{2}$ and roots
satisfying $\left\vert v\right\vert ^{2}=2$. The roots in $R_{+}$ all satisfy
$\left\vert v\right\vert ^{2}=4$. A change of variables leads to the stated formula.
\end{proof}

The integral of the square of the alternating polynomial $a_{G}\left(
x\right)  :=%
{\displaystyle\prod_{v\in R_{+}}}
\left\langle x,v\right\rangle $, namely $a_{G}\left(  \frac{1}{2\omega}%
\nabla_{\kappa}\right)  a_{G}\left(  x\right)  $, can be done with symbolic
computation (in a few minutes of CPU\ time) to show that%
\[
\frac{c_{\kappa,\omega}}{c_{\kappa+1,\omega}}=\frac{120}{\omega^{15}}\left(
2\kappa+1\right)  \left(  6\kappa+1\right)  _{5}\left(  10\kappa+1\right)
_{9}.
\]

Proposition \ref{wHaw} applies to $\mathcal{I}_{+}$ since $\mathcal{I}$ is a
$G$-orbit and this motivates the following:

\begin{definition}
For $k=1,2,3,\ldots$let $H^{\left(  k\right)  }:=\sum_{y\in\mathcal{I}_{+}%
}H_{y}^{k}$.
\end{definition}

Recall $\mathcal{I}_{+}=\left\{  \left(  0,\tau,\pm1\right)  ,\left(
1,0,\pm\tau\right)  ,\left(  \tau,\pm1,0\right)  \right\}  $. By Propositions
(\ref{[hH'0}) and (\ref{wHaw}) $\left[  \mathcal{H},H^{\left(  k\right)
}\right]  =0$ and $\left[  w,H^{\left(  k\right)  }\right]  =0$ for
$k=1,2,\ldots$ and $w\in G$. In terms of action on polynomials%
\begin{equation}
\widetilde{H}^{\left(  k\right)  }=g^{-1}H^{\left(  k\right)  }g=\sum
_{y\in\mathcal{I}_{+}}\left\{  \omega\left(  \left\langle y,x\right\rangle
\left\langle y,\nabla_{\kappa}\right\rangle +\left\langle y,\nabla_{\kappa
}\right\rangle \left\langle y,x\right\rangle \right)  -\left\langle
y,\nabla_{\kappa}\right\rangle ^{2}\right\}  ^{k}\label{Hkdefn}%
\end{equation}
Thus $\widetilde{H}^{\left(  k\right)  }$ maps $\mathcal{P}^{G}$ to
$\mathcal{P}^{G}$. The equation $\sum_{y\in\mathcal{I}_{+}}\left\langle
x,y\right\rangle ^{2}=2\left(  \tau+2\right)  \left\vert x\right\vert ^{2}$
shows that $H^{\left(  1\right)  }=2\left(  \tau+2\right)  \mathcal{H}$. There
is a functional relation among $H^{\left(  2\right)  },\mathcal{H}%
,\mathcal{J}$ and $G$. To state this relation we use%
\begin{align*}
G_{\mathrm{rot}} &  =\left\{  w\in G:\det w=1,w\neq I\right\}  \\
\rho_{2} &  =\left\{  w\in G_{\mathrm{rot}}:w^{2}=I\right\}  ,\\
\rho_{3} &  =\left\{  w\in G_{\mathrm{rot}}:w^{3}=I\right\}  ,\\
\rho_{5,1} &  =\left\{  w\in G_{\mathrm{rot}}:w^{5}=I,\mathrm{tr}%
~w=\tau\right\}  \\
\rho_{5,2} &  =\left\{  w\in G_{\mathrm{rot}}:w^{5}=I,\mathrm{tr~}%
w=1-\tau\right\}
\end{align*}
Then $\#G_{\mathrm{rot}}=59$, $\rho_{3}$ consists of 20 rotations fixing an
antipodal pair of vertices of $\mathcal{K}$, $\rho_{5,1}$, $\rho_{5,2}$ each
consist of 12 rotations of $\frac{2\pi}{5}$, respectively $\frac{4\pi}{5}$,
fixing an antipodal pair of vertices of $\mathcal{I}$, and $\rho_{2}$ consists
of 15 rotations about the midpoints of the edges of $\mathcal{I}$. By a
symbolic computer calculation (see Appendix)
\begin{align}
H^{\left(  2\right)  } &  =6\left(  \tau+1\right)  \mathcal{H}^{2}+8\omega
^{2}\left(  \tau+1\right)  \mathcal{J}-24\omega^{2}\left(  \tau+1\right)
\label{H2action}\\
&  -32\omega^{2}\left(  \tau+1\right)  \kappa\sum_{v\in R_{+}}\sigma
_{v}-32\omega^{2}\left(  \tau+1\right)  \kappa^{2}\sum_{w\in\rho_{2}%
}w-36\omega^{2}\left(  \tau+1\right)  \kappa^{2}\sum_{w\in\rho_{3}%
}w\nonumber\\
&  -20\omega^{2}\left(  \tau+2\right)  \kappa^{2}\sum_{w\in\rho_{5,1}%
}w-20\omega^{2}\left(  4\tau+3\right)  \kappa^{2}\sum_{w\in\rho_{5,2}%
}w.\nonumber
\end{align}
On $\mathcal{P}^{G}$ the group terms reduce to the constant $-24\omega
^{2}\left(  \tau+1\right)  \left(  10\kappa+1\right)  ^{2}$. Note that the
group terms are sums over cosets, hence are central in the group algebra of
$G$. Since $\mathcal{J}$ commutes with the group action the above relation
shows that $\left[  H^{\left(  2\right)  },\mathcal{J}\right]  =0$.

To produce an operator commuting with $\mathcal{H}$ and $G$ which is
functionally independent of $\mathcal{H},\mathcal{J}$ we step up to
$H^{\left(  3\right)  }$. A relatively quick symbolic computation at the
$\kappa=0$ level shows that $\left[  H^{\left(  3\right)  },\mathcal{J}%
\right]  \neq0,$establishing the independence. (A similar computation, again
at $\kappa=0$, shows that $\left[  H^{\left(  3\right)  },H^{\left(  5\right)
}\right]  \neq0$). This is heuristically plausible since the lowest degree
mutually independent $G$-invariant polynomials are of degree $2,6$,$10$ and so
are $\mathcal{H},$ $H^{\left(  3\right)  }$ and $H^{\left(  5\right)  }$.
Recall that $H^{\left(  k\right)  }$ is self-adjoint in $L^{2}\left(
\mathbb{R}^{3},h_{\kappa}^{2}\mathrm{d}m\right)  .$

\section{\label{GenFunWF}Explicit icosahedral wavefunctions}

\subsection{Definition and basic properties}

The polynomials are defined by means of a generating function.

\begin{definition}
for $y_{0}\in\mathcal{I}$ and $x\in\mathbb{R}^{3}$ let
\[
F\left(  r,x;y_{0}\right)  =\left(  1-r\left\langle x,y_{0}\right\rangle
\right)  ^{-1}\prod\limits_{y\in\mathcal{I}}\left(  1-r\left\langle
x,y\right\rangle \right)  ^{-\kappa}=\sum_{n=0}^{\infty}q_{n}\left(
x;y_{0}\right)  r^{n}.
\]

\end{definition}

Let $F_{0}\left(  r;x,y\right)  =\prod\limits_{y\in\mathcal{I}}\left(
1-r\left\langle x,y\right\rangle \right)  ^{-\kappa}=:\sum_{n=0}^{\infty
}p_{2n}\left(  x\right)  r^{2n}$ (because $y\in\mathcal{I}$ implies
$-y\in\mathcal{I}$) and thus $q_{n}\left(  x;y_{0}\right)  =\sum
_{j=0}^{\left\lfloor n/2\right\rfloor }\left\langle x,y_{0}\right\rangle
^{n-2j}p_{2j}\left(  x\right)  $. Also $q_{n}\left(  x;-y_{0}\right)
=q_{n}\left(  -x;y_{0}\right)  =\left(  -1\right)  ^{n}q_{n}\left(
x;y_{0}\right)  .$

\begin{theorem}
Suppose $u\in\mathbb{R}^{3}$ and $y_{0}\in\mathcal{I}$ then
\[
\left\langle u,\nabla_{\kappa}\right\rangle F\left(  r,x;y_{0}\right)
=\left\langle u,y_{0}\right\rangle \left(  \left(  r^{2}\frac{\partial
}{\partial r}+r\left(  1+11\kappa\right)  \right)  F\left(  r,x;y_{0}\right)
-\kappa rF\left(  -r,x;y_{0}\right)  \right)  .
\]

\end{theorem}

\begin{proof}
The product $F_{0}$ is $G$-invariant and the product rule for $\nabla_{\kappa
}$ applies. Let $f_{0}\left(  r,x;y_{0}\right)  =\left(  1-r\left\langle
x,y_{0}\right\rangle \right)  ^{-1}$. By use of the logarithmic derivative%
\[
\frac{1}{F}\left\langle u,\nabla_{\kappa}\right\rangle F=\frac{\left\langle
u,\nabla\right\rangle F_{0}}{F_{0}}+\frac{\left\langle u,\nabla\right\rangle
f_{0}}{f_{0}}+\kappa\sum_{v\in R_{+}}\frac{f_{0}\left(  r,x;y_{0}\right)
-f_{0}\left(  r,x\sigma_{v};y_{0}\right)  }{\left\langle x,v\right\rangle
f_{0}\left(  r,x;y_{0}\right)  }\left\langle u,v\right\rangle .
\]
The first two terms are%
\[
\frac{\left\langle u,\nabla\right\rangle F_{0}}{F_{0}}+\frac{\left\langle
u,\nabla\right\rangle f_{0}}{f_{0}}=\sum_{y\in\mathcal{I}}\frac{r\kappa
\left\langle u,y\right\rangle }{1-r\left\langle x,y\right\rangle }%
+\frac{r\left\langle u,y_{0}\right\rangle }{1-r\left\langle x,y_{0}%
\right\rangle }%
\]
and the sum over $R_{+}$ is
\begin{gather*}
\kappa\sum_{v\in R_{+}}\frac{f_{0}\left(  r,x;y_{0}\right)  -f_{0}\left(
r,x\sigma_{v};y_{0}\right)  }{\left\langle x,v\right\rangle f_{0}\left(
r,x;y_{0}\right)  }\left\langle u,v\right\rangle =\kappa\sum_{v\in R_{+}}%
\frac{\left\langle u,v\right\rangle }{\left\langle x,v\right\rangle }\left\{
1-\frac{1-r\left\langle x,y_{0}\right\rangle }{1-r\left\langle x\sigma
_{v},y_{0}\right\rangle }\right\} \\
=\kappa\sum_{v\in R_{+}}\frac{r\left\langle u,v\right\rangle }{\left\langle
x,v\right\rangle }\frac{\left\langle x,y_{0}\right\rangle -\left\langle
x\sigma_{v},y_{0}\right\rangle }{1-r\left\langle x\sigma_{v},y_{0}%
\right\rangle }=\kappa\sum_{v\in R_{+}}\frac{2r}{\left\vert v\right\vert ^{2}%
}\frac{\left\langle u,v\right\rangle \left\langle y_{0},v\right\rangle
}{1-r\left\langle x\sigma_{v},y_{0}\right\rangle }.
\end{gather*}
The five reflections that fix $y_{0}$ do not appear in the sum. For each
$y\in\mathcal{I}\backslash\left\{  \pm y_{0}\right\}  $ there is a unique
$v\in R_{+}$ such that $y_{0}\sigma_{v}=y$; thus the term $\dfrac
{1}{1-r\left\langle x,y\right\rangle }$ appears in $\frac{1}{F}\left\langle
u,\nabla_{\kappa}\right\rangle F$ with coefficient $r\kappa\left\{
\left\langle u,y\right\rangle +\frac{2}{\left\vert v\right\vert ^{2}%
}\left\langle u,v\right\rangle \left\langle y_{0},v\right\rangle \right\}
=r\kappa\left\{  \left\langle u,y_{0}\sigma_{v}\right\rangle +\frac
{2}{\left\Vert v\right\Vert ^{2}}\left\langle u,v\right\rangle \left\langle
y_{0},v\right\rangle \right\}  =r\kappa\left\langle u,y_{0}\right\rangle $ (by
$\left\langle x\sigma_{v},y_{0}\right\rangle =\left\langle x,y_{0}\sigma
_{v}\right\rangle $ and the definition of $y_{0}\sigma_{v}$). The terms
$\dfrac{1}{1-r\left\langle x,y_{0}\right\rangle }$ and $\dfrac{1}%
{1+r\left\langle x,y_{0}\right\rangle }$ appear with coefficients $r\left(
\kappa+1\right)  \left\langle u,y_{0}\right\rangle $ and $-r\kappa\left\langle
u,y_{0}\right\rangle $ respectively. Thus%
\begin{align*}
\frac{1}{F}\left\langle u,\nabla_{\kappa}\right\rangle F  &  =\kappa
r\left\langle u,y_{0}\right\rangle \sum_{y\in\mathcal{I}\backslash\left\{  \pm
y_{0}\right\}  }\left\{  1+\frac{r\left\langle x,y\right\rangle }%
{1-r\left\langle x,y\right\rangle }\right\}  +r\left(  \kappa+1\right)
\left\langle u,y_{0}\right\rangle \left\{  1+\frac{r\left\langle
x,y_{0}\right\rangle }{1-r\left\langle x,y_{0}\right\rangle }\right\} \\
&  -\left\langle u,y_{0}\right\rangle \left\{  \frac{\kappa r\left\langle
x,y_{0}\right\rangle }{1+r\left\langle x,y_{0}\right\rangle }+\kappa
\frac{1-r\left\langle x,y_{0}\right\rangle }{1+r\left\langle x,y_{0}%
\right\rangle }\right\} \\
&  =\frac{\left\langle u,y_{0}\right\rangle }{F}\left\{  \left(  r^{2}%
\frac{\partial}{\partial r}+r\left(  1+11\kappa\right)  \right)  F\left(
r,x;y_{0}\right)  -\kappa rF\left(  -r,x;y_{0}\right)  \right\}  ;
\end{align*}
Note that $\#\left(  \mathcal{I}\backslash\left\{  \pm y_{0}\right\}  \right)
=10$ and $\dfrac{F\left(  -r,x;y_{0}\right)  }{F\left(  r,x;y_{0}\right)
}=\dfrac{1-r\left\langle x,y_{0}\right\rangle }{1+r\left\langle x,y_{0}%
\right\rangle }.$This completes the proof.
\end{proof}

\begin{corollary}
$\left\langle u,\nabla_{\kappa}\right\rangle q_{2n}\left(  x;y_{0}\right)
=2\left\langle u,y_{0}\right\rangle \left(  6\kappa+n\right)  q_{2n-1}\left(
x;y_{0}\right)  $ and $\left\langle u,\nabla_{\kappa}\right\rangle
q_{2n+1}\left(  x;y_{0}\right)  =\left\langle u,y_{0}\right\rangle \left(
10\kappa+2n+1\right)  q_{2n}\left(  x;y_{0}\right)  $.
\end{corollary}

\begin{proof}%
\begin{align*}
\sum_{n=0}^{\infty}r^{n}\left\langle u,\nabla_{\kappa}\right\rangle
q_{n}\left(  x;y_{0}\right)   &  =\left\langle u,y_{0}\right\rangle \left(
\begin{array}
[c]{c}%
\left(  r^{2}\frac{\partial}{\partial r}+r\left(  1+11\kappa\right)  \right)
\sum_{n=0}^{\infty}r^{n}q_{n}\left(  x;y_{0}\right)  \\
-\kappa r\sum_{n=0}^{\infty}\left(  -r\right)  ^{n}q_{n}\left(  x;y_{0}%
\right)
\end{array}
\right)  \\
&  =\left\langle u,y_{0}\right\rangle \sum_{n=0}^{\infty}\left\{
n+1+11\kappa+\left(  -1\right)  ^{n+1}\kappa\right\}  r^{n+1}q_{n}\left(
x;y_{0}\right)  .
\end{align*}

\end{proof}

This is an attractive formula; in general the action of $\left\langle
u,\nabla_{\kappa}\right\rangle $ is quite messy - for example $\mathcal{D}%
_{1}x_{1}^{3}=\left(  3+\frac{23}{2}\kappa\right)  x_{1}^{2}-\frac{\kappa}%
{2}\left(  \tau-7\right)  x_{2}^{2}+\frac{\kappa}{2}\left(  \tau+6\right)
x_{3}^{2}$. In contrast let $y_{0}=\left(  0,\tau,1\right)  $ then
$q_{2}\left(  x;y_{0}\right)  =\left(  \tau x_{2}+x_{3}\right)  ^{2}+2\left(
\tau+2\right)  \left\vert x\right\vert ^{2}$, $q_{3}\left(  x;y_{0}\right)
=\left(  \tau x_{2}+x_{3}\right)  q_{2}\left(  x;y_{0}\right)  $ and
$\left\langle u,\nabla_{\kappa}\right\rangle q_{3}\left(  x;y_{0}\right)
=\left(  \tau u_{2}+u_{3}\right)  \left(  10\kappa+3\right)  q_{2}\left(
x;y_{0}\right)  $.

Recall (see \cite[Sect. 6.5]{DX2014}) the intertwining operator $V$ satisfying
$\left\langle u,\nabla_{\kappa}\right\rangle Vp\left(  x\right)  =V\left(
\left\langle u,\nabla\right\rangle p\left(  x\right)  \right)  $, $V1=1$ and
$p\in\mathcal{P}_{n}$ implies $Vp\in\mathcal{P}_{n}$.

\begin{definition}
For $n=0,1,2,\ldots$let $\nu\left(  n\right)  =2^{n}\left(  6\kappa+1\right)
_{s}\left(  5\kappa+\frac{1}{2}\right)  _{t}$ with $s=\left\lfloor \frac{n}%
{2}\right\rfloor ,\,t=\left\lfloor \frac{n+1}{2}\right\rfloor $.
\end{definition}

\begin{proposition}
\label{Vofq}Suppose $y_{0}\in\mathcal{I}$ then $V\left(  \left\langle
x,y_{0}\right\rangle ^{n}\right)  =\dfrac{n!}{\nu\left(  n\right)  }%
q_{n}\left(  x;y_{0}\right)  $.
\end{proposition}

\begin{proof}
Proceed by induction. The formula is trivially true for $n=0.$ Suppose $n$ is
odd and the formula is true for $n-1$, then%
\begin{align*}
\left\langle u,\nabla_{\kappa}\right\rangle \frac{n!}{\nu\left(  n\right)
}q_{n}\left(  x;y_{0}\right)   &  =\left\langle u,y_{0}\right\rangle \frac
{n!}{\nu\left(  n\right)  }\left(  10\kappa+n\right)  q_{n-1}\left(
x;y_{0}\right) \\
\left\langle u,\nabla_{\kappa}\right\rangle \left(  V\left\langle
x,y_{0}\right\rangle ^{n}\right)   &  =V\left(  \left\langle u,\nabla
\right\rangle \left\langle x,y_{0}\right\rangle ^{n}\right)  =n\left\langle
u,y_{0}\right\rangle V\left\langle x,y_{0}\right\rangle ^{n-1}\\
&  =n\left\langle u,y_{0}\right\rangle \frac{\left(  n-1\right)  !}{\nu\left(
n-1\right)  }q_{n-1}\left(  x;y_{0}\right)  ,
\end{align*}
with $\dfrac{10\kappa+n}{\nu\left(  n\right)  }=\dfrac{1}{\nu\left(
n-1\right)  }.$ Suppose $n$ is even and the formula is true for $n-1$, then%
\begin{align*}
\left\langle u,\nabla_{\kappa}\right\rangle \frac{n!}{\nu\left(  n\right)
}q_{n}\left(  x;y_{0}\right)   &  =\left\langle u,y_{0}\right\rangle \frac
{n!}{\nu\left(  n\right)  }\left(  12\kappa+n\right)  q_{n-1}\left(
x;y_{0}\right) \\
\left\langle u,\nabla_{\kappa}\right\rangle \left(  V\left\langle
x,y_{0}\right\rangle ^{n}\right)   &  =V\left(  \left\langle u,\nabla
\right\rangle \left\langle x,y_{0}\right\rangle ^{n}\right)  =n\left\langle
u,y_{0}\right\rangle V\left\langle x,y_{0}\right\rangle ^{n-1}\\
&  =n\left\langle u,y_{0}\right\rangle \frac{\left(  n-1\right)  !}{\nu\left(
n-1\right)  }q_{n-1}\left(  x;y_{0}\right)  ,
\end{align*}
and $\dfrac{12\kappa+n}{\nu\left(  n\right)  }=\dfrac{1}{\nu\left(
n-1\right)  }$.
\end{proof}

\begin{corollary}
\label{deltam}Suppose $y_{0}\in\mathcal{I}$ then $\Delta_{\kappa}q_{n}\left(
x;y_{0}\right)  =\dfrac{\nu\left(  n\right)  }{\nu\left(  n-2\right)  }\left(
\tau+2\right)  q_{n-2}\left(  x;y_{0}\right)  $, and $\Delta_{\kappa}^{m}%
q_{n}\left(  x;y_{0}\right)  =\dfrac{\nu\left(  n\right)  }{\nu\left(
n-2m\right)  }\left(  \tau+2\right)  ^{m}q_{n-2m}\left(  x;y_{0}\right)  $.
\end{corollary}

\begin{proof}
Note $\Delta\left\langle x,y\right\rangle ^{n}=n\left(  n-1\right)  \left\vert
y\right\vert ^{2}\left\langle x,y\right\rangle ^{n-2}$ and $\left\vert
y_{0}\right\vert ^{2}=\tau^{2}+1=\tau+2$. We have%
\begin{align*}
\Delta_{\kappa}q_{n}\left(  x;y_{0}\right)   &  =\frac{\nu\left(  n\right)
}{n!}\Delta_{\kappa}V\left(  \left\langle x,y_{0}\right\rangle ^{n}\right)
=\frac{\nu\left(  n\right)  }{n!}V\left(  \Delta\left\langle x,y_{0}%
\right\rangle ^{n}\right) \\
&  =\frac{\nu\left(  n\right)  }{n!}n\left(  n-1\right)  \left(
\tau+2\right)  V\left(  \left\langle x,y_{0}\right\rangle ^{n-2}\right) \\
&  =\frac{\nu\left(  n\right)  }{n!}\left(  \tau+2\right)  \frac{n\left(
n-1\right)  \left(  n-2\right)  !}{\nu\left(  n-2\right)  }q_{n-2}\left(
x;y_{0}\right)  .
\end{align*}

\end{proof}

The Corollary provides the formulas need to find wavefunctions with highest
degree part $q_{n}\left(  x;y_{0}\right)  $ and to produce harmonic
polynomials. Note%
\begin{align}
\frac{\nu\left(  2n\right)  }{\nu\left(  2n-2m\right)  }  & =2^{2m}\left(
-6\kappa-n\right)  _{m}\left(  -5\kappa+\frac{1}{2}-n\right)  _{m}%
\label{nu(2n)/(2n-2j)}\\
\frac{\nu\left(  2n+1\right)  }{\nu\left(  2n+1-2m\right)  }  & =2^{2m}\left(
-6\kappa-n\right)  _{m}\left(  -5\kappa-\frac{1}{2}-n\right)  _{m}%
.\label{vu(2n+1)(-2j)}%
\end{align}

\subsection{Wavefunctions and harmonic polynomials}

From Proposition \ref{expwp} $\left(  \exp\left(  -\frac{\Delta_{\kappa}%
}{4\omega}\right)  q_{n}\left(  x;y_{0}\right)  \right)  g$ is a wavefunction
with highest degree term being $q_{n}\left(  x;y_{0}\right)  $. For the group
$G$ the energy eigenvalue specializes to $E_{n}=\omega\left(  30\kappa
+2n+3\right)  $.

\begin{definition}
For $y_{0}\in\mathcal{I}$ and $n=1,2,3,\ldots$ let%
\[
w_{n}\left(  x;y_{0}\right)  :=%
{\displaystyle\sum\limits_{j=0}^{\left\lfloor n/2\right\rfloor }}
\left(  -1\right)  ^{j}\dfrac{\left(  \tau+2\right)  ^{j}}{\left(
4\omega\right)  ^{j}j!}\dfrac{\nu\left(  n\right)  }{\nu\left(  n-2j\right)
}q_{n-2j}\left(  x;y_{0}\right)  .
\]

\end{definition}

By Corollary \ref{deltam} $\Delta_{\kappa}^{j}w_{n}\left(  x;y_{0}\right)
=\dfrac{\nu\left(  n\right)  }{\nu\left(  n-2j\right)  }\left(  \tau+2\right)
^{j}w_{n-2j}\left(  x;y_{0}\right)  .$and thus $\widetilde{\mathcal{H}}%
w_{n}\left(  x;y_{0}\right)  =\allowbreak\omega\left(  15\kappa+2n+3\right)
w_{n}\left(  x;y_{0}\right)  .$ We use Theorem \ref{wavenorm} to find
$\left\langle w_{n}\left(  \cdot;y_{0}\right)  ,w\left(  \cdot;y_{1}\right)
\right\rangle _{2}$.

\begin{lemma}
\label{pdelq}Suppose $p\in\mathcal{P}_{n}$ and $y_{0}\in\mathcal{I}$ then
$\left\langle p,q_{n}\left(  \cdot;y_{0}\right)  \right\rangle _{\kappa
,\omega}=\left(  2\omega\right)  ^{-n}\nu\left(  n\right)  p\left(
y_{0}\right)  $.
\end{lemma}

\begin{proof}
By Proposition \ref{Vofq} $q_{n}\left(  x;y_{0}\right)  =\dfrac{\nu\left(
n\right)  }{n!}V\left(  \left\langle x,y_{0}\right\rangle ^{n}\right)  $ thus%
\begin{align*}
p\left(  \nabla_{\kappa}\right)  q_{n}\left(  x;y_{0}\right)   &  =\frac
{\nu\left(  n\right)  }{n!}p\left(  \nabla_{\kappa}\right)  V\left(
\left\langle x,y_{0}\right\rangle ^{n}\right)  =\frac{\nu\left(  n\right)
}{n!}V\left(  p\left(  \nabla\right)  \left\langle x,y_{0}\right\rangle
^{n}\right) \\
&  =\frac{\nu\left(  n\right)  }{n!}\left\langle y_{0},\nabla\right\rangle
^{n}p\left(  x\right)  =\nu\left(  n\right)  p\left(  y_{0}\right)  .
\end{align*}
This used $V1=1$ and $p,q\in\mathcal{P}_{n}$ implies $p\left(  \nabla\right)
q\left(  x\right)  =q\left(  \nabla\right)  p\left(  x\right)  $.
\end{proof}

By Theorem \ref{wavenorm} and the Lemma
\begin{equation}
\left\langle w_{n}\left(  \cdot;y_{0}\right)  ,w_{n}\left(  \cdot
;y_{1}\right)  \right\rangle _{2}=\left\langle q_{n}\left(  \cdot
;y_{1}\right)  ,q_{n}\left(  \cdot;y_{0}\right)  \right\rangle _{\kappa
,\omega}=\left(  2\omega\right)  ^{-n}\nu\left(  n\right)  q_{n}\left(
y_{1};y_{0}\right)  . \label{winpro}%
\end{equation}
There are two cases for $\left\{  y_{0},y_{1}\right\}  $: (1) $y_{1}=\pm
y_{0}$ and $\left\langle y_{0},y_{1}\right\rangle =\pm\left(  \tau+2\right)
$; (2) $\left\langle y_{0},y_{1}\right\rangle =\pm\tau$. The list $\left[
\left\langle y,y_{1}\right\rangle :y\in\mathcal{I}_{+}\right]  $ comprises
$\left(  \tau+2\right)  ,-\left(  \tau+2\right)  ,$ 5 occurrences each of
$\tau,-\tau$ thus $\prod\limits_{y\in\mathcal{I}}\left(  1-r\left\langle
y_{1},y\right\rangle \right)  ^{-\kappa}=\left(  1-5\tau^{2}r^{2}\right)
^{-\kappa}\left(  1-\tau^{2}r^{2}\right)  ^{-5\kappa}$ (since $\left(
\tau+2\right)  ^{2}=5\tau^{2}$). Because $q_{n}\left(  x;-y_{1}\right)
=\left(  -1\right)  ^{n}q_{n}\left(  x;y_{1}\right)  $ it suffices to consider
$\left\langle y_{0},y_{1}\right\rangle =\tau+2,\tau$. If $y_{1}=y_{0}$ then%
\begin{align*}
\sum_{n=0}^{\infty}q_{n}\left(  y_{0};y_{0}\right)  r^{n}  &  =\left(
1-\left(  \tau+2\right)  r\right)  ^{-1}\left(  1-5\tau^{2}r^{2}\right)
^{-\kappa}\left(  1-\tau^{2}r^{2}\right)  ^{-5\kappa}\\
&  =\left(  1+\left(  \tau+2\right)  r\right)  \left(  1-5\tau^{2}%
r^{2}\right)  ^{-\kappa-1}\left(  1-\tau^{2}r^{2}\right)  ^{-5\kappa},
\end{align*}
and if $\left\langle y_{0},y_{1}\right\rangle =\tau$ then%
\begin{align*}
\sum_{n=0}^{\infty}q_{n}\left(  y_{0};y_{1}\right)  r^{n}  &  =\left(  1-\tau
r\right)  ^{-1}\left(  1-5\tau^{2}r^{2}\right)  ^{-\kappa}\left(  1-\tau
^{2}r^{2}\right)  ^{-5\kappa}\\
&  =\left(  1+\tau r\right)  \left(  1-5\tau^{2}r^{2}\right)  ^{-\kappa
}\left(  1-\tau^{2}r^{2}\right)  ^{-5\kappa-1}.
\end{align*}

\begin{definition}
For $n=0,1,2,\ldots$let%
\[
Y_{n}^{\left(  0\right)  }=\sum_{j=0}^{n}\frac{\left(  \kappa+1\right)
_{j}\left(  5\kappa\right)  _{n-j}}{j!\left(  n-j\right)  !}5^{j}%
,~Y_{n}^{\left(  1\right)  }=\sum_{j=0}^{n}\frac{\left(  \kappa\right)
_{j}\left(  5\kappa+1\right)  _{n-j}}{j!\left(  n-j\right)  !}5^{j}.
\]

\end{definition}

By the negative binomial series $q_{2n}\left(  y_{0};y_{0}\right)  =\tau
^{2n}Y_{n}^{\left(  0\right)  },q_{2n+1}\left(  y_{0};y_{0}\right)  =\tau
^{2n}\left(  \tau+2\right)  Y_{n}^{\left(  0\right)  }$, and if $\left\langle
y_{0},y_{1}\right\rangle =\tau$ then $q_{2n}\left(  y_{0};y_{1}\right)
=\tau^{2n}Y_{n}^{\left(  1\right)  }$, and $q_{2n+1}\left(  y_{0}%
;y_{1}\right)  =\tau^{2n+1}Y_{n}^{\left(  1\right)  }$. By equation
(\ref{winpro})%
\begin{align}
\left\Vert w_{2n}\left(  \cdot;y_{0}\right)  \right\Vert _{2}^{2} &  =\left(
2\omega\right)  ^{-2n}\nu\left(  2n\right)  \tau^{2n}Y_{n}^{\left(  0\right)
},\label{w2norm}\\
\left\Vert w_{2n+1}\left(  \cdot;y_{0}\right)  \right\Vert _{2}^{2} &
=\left(  2\omega\right)  ^{-2n-1}\nu\left(  2n+1\right)  \left(
\tau+2\right)  \tau^{2n}Y_{n}^{\left(  0\right)  },\nonumber
\end{align}
and if $\left\langle y_{0},y_{1}\right\rangle =\tau$ then
\begin{align}
\left\langle w_{2n}\left(  \cdot;y_{0}\right)  ,w_{2n}\left(  \cdot
;y_{1}\right)  \right\rangle _{2} &  =\left(  2\omega\right)  ^{-2n}\nu\left(
2n\right)  \tau^{2n}Y_{n}^{\left(  1\right)  }\label{w2ip}\\
\left\langle w_{2n+1}\left(  \cdot;y_{0}\right)  ,w_{2n+1}\left(  \cdot
;y_{1}\right)  \right\rangle _{2} &  =\left(  2\omega\right)  ^{-2n}\nu\left(
2n+1\right)  \tau^{2n+1}Y_{n}^{\left(  1\right)  }.\nonumber
\end{align}

Harmonic polynomials are constructed by means of Proposition \ref{hmproj}
specialized to $N=3,\gamma_{\kappa}=15\kappa$.

\begin{definition}
\label{defH3phi}For $y_{0}\in\mathcal{I}$ and $n=1,2,3,\ldots$ let
\[
\phi_{n}\left(  x;y_{0}\right)  =%
{\displaystyle\sum\limits_{j=0}^{\left\lfloor n/2\right\rfloor }}
\dfrac{\left(  \tau+2\right)  ^{j}\left\vert x\right\vert ^{2j}}%
{4^{j}j!\left(  -15\kappa-n+1/2\right)  _{j}}\dfrac{\nu\left(  n\right)  }%
{\nu\left(  n-2j\right)  }q_{n-2j}\left(  x;y_{0}\right)  .
\]

\end{definition}

By Corollary \ref{deltam} $\Delta_{\kappa}\phi_{n}\left(  \cdot;y_{0}\right)
=0$ and $\phi_{n}\left(  \cdot;y_{0}\right)  g$ is a wavefunction (energy
eigenvalue $E_{n}$). The inner product $\left\langle \phi_{n}\left(
\cdot;y_{0}\right)  ,\phi_{n}\left(  \cdot;y_{1}\right)  \right\rangle
_{2}=\left\langle \phi_{n}\left(  \cdot;y_{1}\right)  ,\phi_{n}\left(
\cdot;y_{0}\right)  \right\rangle _{\kappa,\omega}$ is more complicated than
the previous case.

\begin{theorem}
For $y_{0},y_{1}\in\mathcal{I}$%
\[
\left\langle \phi_{n}\left(  \cdot;y_{0}\right)  ,\phi_{n}\left(  \cdot
;y_{1}\right)  \right\rangle _{2}=\sum_{j=0}^{\left\lfloor n/2\right\rfloor
}\dfrac{\left(  2\omega\right)  ^{-n}\left(  \tau+2\right)  ^{2j}}%
{4^{j}j!\left(  -15\kappa-n+1/2\right)  _{j}}\dfrac{\nu\left(  n\right)  ^{2}%
}{\nu\left(  n-2j\right)  }q_{n-2j}\left(  y_{0};y_{1}\right)  .
\]

\end{theorem}

Abbreviate $d_{n,j}=\dfrac{\left(  \tau+2\right)  ^{j}}{4^{j}j!\left(
-15\kappa-n+1/2\right)  _{j}}\dfrac{\nu\left(  n\right)  }{\nu\left(
n-2j\right)  }$. Then%
\begin{gather*}
\left(  2\omega\right)  ^{n}\left\langle \phi_{n}\left(  \cdot;y_{0}\right)
,\phi_{n}\left(  \cdot;y_{1}\right)  \right\rangle _{\kappa,\omega}=\sum
_{j=0}^{\left\lfloor n/2\right\rfloor }d_{n,j}\Delta_{\kappa}^{j}%
q_{n-2j}\left(  \nabla_{\kappa};y_{0}\right)  \phi_{n}\left(  x;y_{1}\right)
\\
=q_{n}\left(  \nabla_{\kappa};y_{0}\right)  \phi_{n}\left(  x;y_{1}\right)
=\nu\left(  n\right)  \phi_{n}\left(  y_{0};y_{1}\right)  =\nu\left(
n\right)  \sum_{j=0}^{\left\lfloor n/2\right\rfloor }d_{n,j}\left\vert
y_{0}\right\vert ^{2j}q_{n-2j}\left(  y_{0};y_{1}\right)  .
\end{gather*}
This used Lemma \ref{pdelq}, $\Delta_{\kappa}^{2j}\phi_{n}\left(  \cdot
;y_{1}\right)  =0$ for $j>0$ and $\left\vert y_{0}\right\vert ^{2}=\tau+2$.
Note $\left\langle \phi_{n}\left(  \cdot;y_{0}\right)  ,\phi_{n}\left(
\cdot;y_{1}\right)  \right\rangle _{2}=\frac{\nu\left(  n\right)  }{\left(
2\omega\right)  ^{n}}\phi_{n}\left(  y_{1};y_{0}\right)  $.

Multiplication by Laguerre polynomials can be applied here. By Propositions
\ref{hmXLag} and \ref{Laghm} the polynomial $f_{n,m}\left(  x\right)
=L_{m}^{\left(  15\kappa+n+1/2\right)  }\left(  \omega\left\vert x\right\vert
^{2}\right)  \phi_{n}\left(  x;y_{0}\right)  $ satisfies $\widetilde
{\mathcal{H}}f_{n,m}=E_{n+2m}f_{n,m}$ and $\left\Vert f_{n,m}\right\Vert
_{2}^{2}=\frac{1}{m!}\left(  \frac{3}{2}+15\kappa+n\right)  _{m}~\left\Vert
\phi_{n}\right\Vert _{2}^{2}$ . Also $\Delta_{\kappa}f_{n,m}=-4\omega\left(
15\kappa+n+m+\frac{1}{2}\right)  f_{n,m-1}$, by Corollary \ref{delLagh}.

\begin{proposition}
\label{wtophi}Suppose $n=2,3,\ldots$and $y\in\mathcal{I}$ then%
\begin{align*}
w_{n}\left(  x;y\right)    & =\sum_{j=0}^{\left\lfloor n/2\right\rfloor
}\left(  -\frac{\tau+2}{4\omega}\right)  ^{j}\frac{1}{\left(  15\kappa
+3/2+n-2j\right)  _{j}}\\
& \times\frac{\nu\left(  n\right)  }{\nu\left(  n-2j\right)  }L_{j}^{\left(
15\kappa+1/2+n-2j\right)  }\left(  \omega\left\vert x\right\vert ^{2}\right)
\phi_{n-2j}\left(  x;y\right)  .
\end{align*}

\end{proposition}

\begin{proof}
Specialize Proposition \ref{wav2laghm} to $N=3,\gamma_{\kappa}=15\kappa$. By
Definition $w_{n}\left(  \cdot;y\right)  =\exp\left(  -\frac{\Delta_{\kappa}%
}{4\omega}\right)  q_{n}\left(  \cdot;y\right)  $ and $\Delta_{\kappa}%
^{j}q_{n}\left(  \cdot;y\right)  =\left(  \tau+2\right)  ^{j}\frac{\nu\left(
n\right)  }{\nu\left(  n-2j\right)  }q_{n-2j}\left(  \cdot;y\right)  $.
Finally $\Lambda_{n-2j}q_{n-2j}\left(  \cdot;y\right)  =\phi_{n-2j}\left(
\cdot;y\right)  $ (by Definition \ref{defH3phi}).
\end{proof}

\subsection{$G$-invariant wavefunctions}

The $G$-invariant polynomials are all of even degree, and some can be obtained
by summing $q_{2n}\left(  \cdot;y\right)  $ over $y\in\mathcal{I}_{+}$.

\begin{definition}
For $n=1,2,3\ldots$ let $q_{2n}^{G}\left(  x\right)  =\sum_{y\in
\mathcal{I}_{+}}q_{2n}\left(  x;y\right)  .$
\end{definition}

In the formula $q_{2n}\left(  x;y_{0}\right)  =$ $\sum_{j=0}^{n}\left\langle
x,y_{0}\right\rangle ^{2n-2j}p_{2j}\left(  x\right)  $ (recall $\sum
_{j=0}^{\infty}p_{2n}\left(  x\right)  r^{2n}=\prod\limits_{y\in
\mathcal{I}_{+}}\left(  1-r^{2}\left\langle x,y\right\rangle ^{2}\right)
^{-\kappa}$) the term $\left\langle x,y_{0}\right\rangle ^{2m}$ is replaced by
$s_{2m}\left(  x\right)  :=\sum_{y\in\mathcal{I}_{+}}\left\langle
x,y\right\rangle ^{2m}$. A short calculation shows that%
\begin{align*}
s_{2m}\left(  x\right)   &  =2\left(  \tau^{2m}+1\right)  \left(  x_{1}%
^{2m}+x_{2}^{2m}+x_{3}^{2m}\right)  \\
&  +2\sum_{j=1}^{m-1}\binom{2m}{2j}\tau^{2j}\left\{  x_{1}^{2j}x_{2}%
^{2m-2j}+x_{2}^{2j}x_{3}^{2m-2j}+x_{3}^{2j}x_{1}^{2m-2j}\right\}  .
\end{align*}
Some low degree examples are%
\begin{align*}
s_{2}\left(  x\right)    & =2\left(  \tau+2\right)  \left\vert x\right\vert
^{2},~s_{4}\left(  x\right)  =6\left(  \tau+2\right)  \left\vert x\right\vert
^{4},\\
s_{6}\left(  x\right)    & =4\left(  4\tau+3\right)  \left\vert x\right\vert
^{6}+6\left(  2\tau-1\right)  \prod\limits_{y\in\mathcal{I}_{+}}\left\langle
x,y\right\rangle ,\\
s_{10}\left(  x\right)    & =5\left(  5\tau+3\right)  \prod\limits_{y\in
\mathcal{K}_{+}}\left\langle x,y\right\rangle +75\left(  3\tau+1\right)
\left\vert x\right\vert ^{4}\prod\limits_{y\in\mathcal{I}_{+}}\left\langle
x,y\right\rangle \\
& +10\left(  11\tau+7\right)  \left\vert x\right\vert ^{10}.
\end{align*}
First consider invariant wavefunctions produced from $q_{2n}^{G}$.

\begin{definition}
For $n=1,2,\ldots$let
\[
w_{2n}^{G}\left(  x\right)  :=%
{\displaystyle\sum\limits_{j=0}^{n}}
\left(  -1\right)  ^{j}\dfrac{\left(  \tau+2\right)  ^{j}}{\left(
4\omega\right)  ^{j}j!}\dfrac{\nu\left(  2n\right)  }{\nu\left(  2n-2j\right)
}q_{2n-2j}^{G}\left(  x\right)  .
\]

\end{definition}

That is $w_{2n}^{G}\left(  x\right)  =\sum_{y\in\mathcal{I}_{+}}w_{2n}\left(
x;y\right)  $ and with $\left\langle y_{0},y_{1}\right\rangle =\tau$%
\begin{align}
\left\Vert w_{2n}^{G}\right\Vert _{2}^{2} &  =\sum_{y,y^{\prime}\in
\mathcal{I}_{+}}\left\langle w_{2n}\left(  \cdot;y\right)  ,w_{2n}\left(
\cdot;y^{\prime}\right)  \right\rangle _{2}\label{wGnorm}\\
&  =6\left\Vert w_{2n}\left(  \cdot;y_{0}\right)  \right\Vert _{2}%
^{2}+30\left\langle w_{2n}\left(  \cdot;y_{0}\right)  ,w_{2n}\left(
\cdot;y_{1}\right)  \right\rangle _{2}\\
&  =6\left(  2\omega\right)  ^{-2n}\nu\left(  2n\right)  \tau^{2n}\left(
Y_{n}^{\left(  0\right)  }+5Y_{n}^{\left(  1\right)  }\right)  \nonumber
\end{align}
from (\ref{w2norm}) and (\ref{w2ip}). There is an identity for polynomials of
hypergeometric type%
\[
\sum_{j=0}^{n}\frac{\left(  \kappa+1\right)  _{j}\left(  5\kappa\right)
_{n-j}}{j!\left(  n-j\right)  !}z^{j}+5\sum_{j=0}^{n}\frac{\left(
\kappa\right)  _{j}\left(  5\kappa+1\right)  _{n-j}}{j!\left(  n-j\right)
!}z^{j}=\left(  6+\frac{n}{\kappa}\right)  \sum_{j=0}^{n}\frac{\left(
\kappa\right)  _{j}\left(  5\kappa\right)  _{n-j}}{j!\left(  n-j\right)
!}z^{j},
\]
and so let $Y_{n}^{\left(  2\right)  }:=$ $\left(  6+\dfrac{n}{\kappa}\right)
\sum\limits_{j=0}^{n}\dfrac{\left(  \kappa\right)  _{j}\left(  5\kappa\right)
_{n-j}}{j!\left(  n-j\right)  !}5^{j}=Y_{n}^{\left(  0\right)  }%
+5Y_{n}^{\left(  1\right)  }$.

Use the same idea on the harmonic polynomials.

For $n=1,2,\ldots$let
\begin{align*}
\phi_{2n}^{G}\left(  x\right)   &  :=%
{\displaystyle\sum\limits_{j=0}^{n}}
\dfrac{\left(  \tau+2\right)  ^{j}\left\vert x\right\vert ^{2j}}%
{4^{j}j!\left(  -15\kappa-2n+1/2\right)  _{j}}\dfrac{\nu\left(  2n\right)
}{\nu\left(  2n-2j\right)  }q_{2n-2j}^{G}\left(  x\right)  \\
&  =\sum_{y\in\mathcal{I}_{+}}\phi_{2n}\left(  x;y\right)  .
\end{align*}
(Thus the expansion in Proposition \ref{wtophi} is valid when $n$ is even and
$w_{n}\left(  \cdot;y\right)  ,\phi_{n-2j}\left(  \cdot;y\right)  $ are
replaced by $w_{n}^{G},\phi_{n-2j}^{G}$ respectively.) Of course for this to
be nonzero it is necessary that $2n\neq2,4,8,14$. In a similar calculation to
(\ref{wGnorm})
\begin{align*}
\left\Vert \phi_{2n}^{G}\right\Vert _{2}^{2} &  =\sum_{y,y^{\prime}%
\in\mathcal{I}_{+}}\left\langle \phi_{2n}\left(  \cdot;y\right)  ,\phi
_{2n}\left(  \cdot;y^{\prime}\right)  \right\rangle _{2}=6\left\Vert \phi
_{2n}\left(  \cdot;y_{0}\right)  \right\Vert _{2}^{2}+30\left\langle \phi
_{2n}\left(  \cdot;y_{0}\right)  ,\phi_{2n}\left(  \cdot;y_{1}\right)
\right\rangle _{2}\\
&  =6\tau^{2n}\frac{\nu\left(  2n\right)  }{\left(  2\omega\right)  ^{2n}}%
\sum_{j=0}^{n}\dfrac{5^{j}\left(  -6\kappa-n\right)  _{j}\left(
-5\kappa-n+1/2\right)  _{j}}{j!\left(  -15\kappa-2n+1/2\right)  _{j}}%
Y_{n-j}^{\left(  2\right)  }.
\end{align*}
The calculation used $\left(  \tau+2\right)  ^{2j}=5^{j}\tau^{2j}$ and the
expression for $\nu(2n)/\nu(2n-2j)$ from (\ref{nu(2n)/(2n-2j)}). Also
$\left\Vert \phi_{2n}^{G}\right\Vert _{2}^{2}=6\frac{\nu\left(  2n\right)
}{\left(  2\omega\right)  ^{2n}}\phi_{2n}^{G}\left(  y\right)  ,y\in
\mathcal{I}$. This formula can be used in symbolic computation for reasonably
small degrees ($\leq50$ or so). It is one of the most effective results of the
paper, since the invariant harmonic polynomials are so fundamental in the
analysis of wavefunctions. Trying direct calculation of $\left\langle
\phi_{2n}^{G},\phi_{2n}^{G}\right\rangle _{\kappa,\omega}$ by symbolic
computation becomes very time and resource demanding for $2n>10$. We find
\begin{align*}
\left\Vert \phi_{6}^{G}\right\Vert _{2}^{2} &  =\left(  2^{6}\times15\right)
\left(  \frac{\tau}{\omega}\right)  ^{6}\left(  6\kappa+1\right)  _{3}\left(
5\kappa+\frac{1}{2}\right)  _{3}\frac{\left(  5\kappa+1\right)  \left(
2\kappa+1\right)  }{30\kappa+7}\\
\left\Vert \phi_{10}^{G}\right\Vert _{2}^{2} &  =\left(  2^{9}\times3\right)
\left(  \frac{\tau}{\omega}\right)  ^{10}\left(  6\kappa+1\right)  _{5}\left(
5\kappa+\frac{1}{2}\right)  _{5}\frac{\left(  5\kappa+1\right)  _{2}\left(
6\kappa+5\right)  }{\left(  30\kappa+11\right)  \left(  30\kappa+17\right)
}\\
\left\Vert \phi_{12}^{G}\right\Vert _{2}^{2} &  =\left(  2^{11}\times
15\right)  \left(  \frac{\tau}{\omega}\right)  ^{12}\left(  6\kappa+1\right)
_{6}\left(  5\kappa+\frac{1}{2}\right)  _{6}\frac{\left(  5\kappa+1\right)
_{3}\left(  \kappa+1\right)  \left(  10\kappa+9\right)  }{\left(
30\kappa+13\right)  \left(  30\kappa+19\right)  \left(  30\kappa+23\right)
}\\
\left\Vert \phi_{16}^{G}\right\Vert _{2}^{2} &  =\left(  2^{16}\times
15\right)  \left(  \frac{\tau}{\omega}\right)  ^{16}\left(  6\kappa+1\right)
_{8}\left(  5\kappa+\frac{1}{2}\right)  _{8}\\
&  \times\frac{\left(  5\kappa+1\right)  _{4}\left(  \kappa+1\right)  \left(
3\kappa+4\right)  }{\left(  30\kappa+17\right)  \left(  30\kappa+23\right)
\left(  30\kappa+27\right)  \left(  30\kappa+29\right)  }.
\end{align*}
Necessarily the sum vanishes for $n=1,2,4,7$. It appears that if $\dim\left(
\mathcal{P}_{2n,\kappa}\cap\mathcal{P}^{G}\right)  =1$ then the squared norm
is a quotient of linear factors in $\kappa$. This does not occur at $2n=30$,
where $\dim\left(  \mathcal{P}_{30,\kappa}\cap\mathcal{P}^{G}\right)  =2$.
Other $G$-invariant wavefunctions can be dealt with: set $f_{n,m}^{G}\left(
x\right)  =L_{m}^{\left(  15\kappa+2n+1/2\right)  }\left(  \omega\left\vert
x\right\vert ^{2}\right)  \phi_{2n}^{G}\left(  x\right)  $ then $\widetilde
{\mathcal{H}}f_{n,m}^{G}=E_{2n+2m}f_{n,m}^{G}$ and $\left\Vert f_{n,m}%
^{G}\right\Vert _{2}^{2}=\dfrac{1}{m!}\left(  15\kappa+2n+\frac{3}{2}\right)
_{m}~\left\Vert \phi_{2n}^{G}\right\Vert _{2}^{2}$ .

\section{\label{2nd6thOps}Second and sixth order operators on wavefunctions}

These are some remarks on the actions of $\mathcal{J}$ and $\widetilde
{H}^{\left(  3\right)  }$ on the $w_{n}$ and $\phi_{n}$ functions. The
polynomials $\mathcal{J}\phi_{n}\left(  \cdot,y\right)  $ have a simple form
(essentially eigenfunctions).

\begin{lemma}
Suppose $\phi\in\mathcal{P}_{n,\kappa}$ then $\mathcal{J}\phi=-\left\langle
x,\nabla_{\kappa}\right\rangle ^{2}\phi-\left\langle x,\nabla_{\kappa
}\right\rangle \phi-2\kappa\sum\limits_{v\in R_{+}}\sigma_{v}\left\langle
x,\nabla_{\kappa}\right\rangle \phi$, and $\left\langle x,\nabla_{\kappa
}\right\rangle \phi=\left(  n+15\kappa\right)  \phi-\kappa\sum\limits_{v\in
R_{+}}\sigma_{v}\phi$.
\end{lemma}

\begin{proof}
This is the specialization of Theorem \ref{angsq} to $N=3$; note
$\Delta_{\kappa}\phi=0$.
\end{proof}

\begin{proposition}
Suppose $y\in\mathcal{I}_{+}$ then%
\[
\mathcal{J}\phi_{2n+1}\left(  \cdot;y\right)  =-2\left(  10\kappa+2n+1\right)
\left(  10\kappa+n+1\right)  \phi_{2n+1}\left(  \cdot;y\right)  .
\]

\end{proposition}

\begin{proof}
The harmonic polynomial $\phi_{2n+1}\left(  \cdot;y\right)  $ is a sum of
$\left\vert x\right\vert ^{2n-2m}q_{2m+1}\left(  \cdot;y\right)  $ over $0\leq
m\leq n$. Since $\sigma_{v}q_{2m+1}\left(  x;y\right)  =q_{2m+1}\left(
x;y\sigma_{v}\right)  $ and 5 reflections fix $y$ and the other 10 map $y$ to
$\mathcal{I}\backslash\left\{  \pm y\right\}  $ it follows that $\sum
\limits_{v\in R_{+}}\sigma_{v}q_{2m+1}\left(  \cdot;y\right)  =5q_{2m+1}%
\left(  \cdot;y\right)  $. By the Lemma $\left\langle x,\nabla_{\kappa
}\right\rangle \phi_{2n+1}\left(  \cdot;y\right)  =\left(  10\kappa
+2n+1\right)  \phi_{2n+1}\left(  \cdot;y\right)  $ and thus%
\begin{align*}
\mathcal{J}\phi_{2n+1}\left(  \cdot;y\right)   &  =-\left(  10\kappa
+2n+1\right)  \left\{  \left(  10\kappa+2n+1\right)  +1+10\kappa\right\}
\phi_{2n+1}\left(  \cdot;y\right)  \\
&  =-2\left(  10\kappa+2n+1\right)  \left(  10\kappa+n+1\right)  \phi
_{2n+1}\left(  \cdot;y\right)  .
\end{align*}

\end{proof}

\begin{proposition}
Suppose $n\geq3$ then $\mathcal{J}\phi_{2n}^{G}=-2n\left(  30\kappa
+2n+1\right)  \phi_{2n}^{G}$.
\end{proposition}

\begin{proof}
Similarly to the previous proof $\phi_{2n}^{G}$ is a sum of $\left\vert
x\right\vert ^{2n-2m}q_{2m}^{G}$; but in this case $\sum\limits_{v\in R_{+}%
}\sigma_{v}\phi_{2n}^{G}=15\phi_{2n}^{G}$ so that $\left\langle x,\nabla
_{\kappa}\right\rangle \phi_{2n}^{G}=2n\phi_{2n}^{G}$ and $\mathcal{J}%
\phi_{2n}^{G}=-\left(  2n\right)  ^{2}-2n-2\left(  15\kappa\right)  \left(
2n\right)  =-2n\left(  30\kappa+2n+1\right)  \phi_{2n}^{G}$.
\end{proof}

\begin{proposition}
Suppose $y\in\mathcal{I}_{+}$ then%
\[
\mathcal{J}\phi_{2n}\left(  \cdot;y_{0}\right)  =-2\left(  6\kappa+n\right)
\left(  18\kappa+2n+1\right)  \phi_{2n}\left(  \cdot;y\right)  +2\kappa\left(
18\kappa+1\right)  \phi_{2n}^{G}.
\]

\end{proposition}

\begin{proof}
In this case $\sum_{v\in R_{+}}\sigma_{v}q_{2m}\left(  x;y\right)
=3q_{2m}\left(  x;y\right)  +2q_{2m}^{G}\left(  x\right)  $ and
\[
\left\langle x,\nabla_{\kappa}\right\rangle \phi_{2n}\left(  \cdot
;y_{0}\right)  =\left(  2n+15\kappa-3\kappa\right)  \phi_{2n}\left(
\cdot;y\right)  -2\kappa\phi_{2n}^{G}.
\]
Substituting these values in the formula for $\mathcal{J}$ and using
$\sum\limits_{v\in R_{+}}\sigma_{v}\phi_{2n}^{G}=15\phi_{2n}^{G}$ gives the
stated result.
\end{proof}

Recall that $\phi_{2n}^{G}=0$ for $2n=2,4,8,14$.

\begin{corollary}
$\mathcal{J}\left(  \phi_{2n}\left(  \cdot;y_{0}\right)  -\frac{1}{6}\phi
_{2n}^{G}\right)  =-2\left(  6\kappa+n\right)  \left(  18\kappa+2n+1\right)
\left(  \phi_{2n}\left(  \cdot;y_{0}\right)  -\frac{1}{6}\phi_{2n}^{G}\right)
.$
\end{corollary}

\begin{remark}
The eigenvalue of $\sum_{v\in R_{+}}\sigma_{v}$ acting on polynomials of
isotype $\chi$, where $\chi$ is an irreducible character of $G$, is
$15\chi\left(  \sigma_{v}\right)  /\chi\left(  I\right)  $. From the character
table of $G$ we find the possible values are $\pm15,\pm5,\pm3,0$. If
$p\in\mathcal{P}_{n,\kappa}$ and $p$ is of isotype $\chi$ (this means that
there is a subspace of $\mathcal{P}_{n,\kappa}$ which is invariant and
irreducible under $G$, and on which $G$ acts corresponding to $\chi$) with
eigenvalue $\lambda_{\chi}$ then $\mathcal{J}p=-\left(  \left(  15-\lambda
_{\chi}\right)  \kappa+n\right)  \left(  \left(  15+\lambda_{\chi}\right)
\kappa+n+1\right)  p$. In particular the reflection representation is realized
on $\mathcal{P}_{1}$ with $\chi\left(  I\right)  =3,\chi\left(  \sigma
_{v}\right)  =1,\lambda_{\chi}=5$.
\end{remark}

The action of $\mathcal{J}$ on $w_{n}\left(  \cdot;y\right)  $ is derived from
Proposition \ref{wtophi} and the fact that $\mathcal{J}$ commutes with any
polynomial in $\left\vert x\right\vert ^{2}$. Here are a few examples
($y\in\mathcal{I}_{+}$)%
\begin{align*}
\mathcal{J}w_{2}\left(  \cdot,y\right)   &  =-6\left(  6\kappa+1\right)
^{2}\phi_{2}\left(  \cdot;y\right)  \\
\mathcal{J}w_{3}\left(  \cdot,y\right)   &  =-4\left(  10\kappa+3\right)
\left(  5\kappa+1\right)  \phi_{3}\left(  \cdot;y\right)  \\
&  +\frac{2\left(  \tau+2\right)  }{5\omega}\left(  10\kappa+3\right)  \left(
10\kappa+1\right)  ^{2}L_{1}^{\left(  15\kappa+3/2\right)  }\left(
\omega\left\vert x\right\vert ^{2}\right)  \phi_{1}\left(  \cdot;y\right)  ,
\end{align*}
which follows from $w_{3}\left(  \cdot;y\right)  =\phi_{3}\left(
\cdot;y\right)  -\frac{\left(  \tau+2\right)  }{5\omega}\left(  10\kappa
+3\right)  \phi_{1}\left(  \cdot;y\right)  $%
\begin{gather*}
\mathcal{J}w_{4}\left(  \cdot;y\right)  =-4\left(  3\kappa+1\right)  \left(
18\kappa+5\right)  \phi_{4}\left(  \cdot;y\right)  \\
+\frac{12\left(  \tau+2\right)  }{\omega\left(  30\kappa+7\right)  }\left(
3\kappa+1\right)  \left(  6\kappa+1\right)  ^{2}\left(  10\kappa+3\right)
L_{1}^{\left(  15\kappa+5/2\right)  }\left(  \omega\left\vert x\right\vert
^{2}\right)  \phi_{2}\left(  \cdot;y\right)  .
\end{gather*}

The operator $\widetilde{H}^{\left(  3\right)  }$ (see (\ref{Hkdefn}))
preserves wavefunctions and commutes with $G$ but the effects appear to be
complicated. At the degree 2 level we find (with $y\in\mathcal{I}_{+}$)%
\begin{align*}
\widetilde{H}^{\left(  3\right)  }\phi_{2}\left(  x;y\right)   &  =2\omega
^{3}\left(  4\tau+3\right)  \left(  25080\kappa^{3}+19772\kappa^{2}%
+5058\kappa+419\right)  \phi_{2}\left(  x;y\right)  ,\\
\widetilde{H}^{\left(  3\right)  }w_{2}\left(  x;y\right)   &  =10\omega
^{3}\left(  4\tau+3\right)  \left(  3000\kappa^{3}+4276\kappa^{2}%
+1386\kappa+127\right)  w_{2}\left(  x;y\right)  \\
&  +48\omega^{3}\left(  4\tau+3\right)  \left(  420\kappa^{3}-67\kappa
^{2}-78\kappa-9\right)  \phi_{2}\left(  x;y\right)  ,\\
\widetilde{H}^{\left(  3\right)  }L_{2}^{\left(  15\kappa+1/2\right)  }\left(
\omega\left\vert x\right\vert ^{2}\right)   &  =2\omega^{3}\left(
4\tau+3\right)  \left(  30\kappa+11\right)  \\
&  \times\left(  500\kappa^{2}+1092\kappa+245\right)  L_{2}^{\left(
15\kappa+1/2\right)  }\left(  \omega\left\vert x\right\vert ^{2}\right)  .
\end{align*}
As mentioned before it is only known that $\widetilde{H}^{\left(  3\right)  }$
does not commute with $\mathcal{J}$.

\section{\label{ConcRem}Concluding Remarks}

The ultimate goal would be to find formulas allowing the construction of
orthogonal bases for wavefunctions and harmonic polynomials of any degree.
However it should be pointed out that this has not yet been done for the
octahedral group, the group of type $B_{3}$, even though transpositions and
sign-changes of variables are easier to deal with than the reflections in the
icosahedral group. The general formulas (Proposition \ref{expwp}) do produce
wavefunctions, for example in the proof of complete integrability of the
Calogero-Moser model of identical particles on a line with harmonic
confinement in $r^{-2}$ interaction potential, but there are no explicit
formulas for the action of $\exp\left(  \frac{\Delta_{\kappa}}{4\omega
}\right)  $ for the symmetric group (type $A_{N-1}$).

One approach might be to analyze monomials $q_{n}\left(  x;y_{1}\right)
q_{m}\left(  x;y_{2}\right)  q_{k}\left(  x;y_{3}\right)  $, but there seems
to be nothing like the manageable generating function $F\left(  r,x;y\right)
$. There are already serious technical difficulties in computing
$\nabla_{\kappa}\left(  q_{n}\left(  \cdot;y_{0}\right)  q_{m}\left(
\cdot;y_{1}\right)  \right)  $ in a usable form. In particular it would be
interesting to construct orthogonal bases for the $G$-invariant harmonic
polynomials, which when multiplied by appropriate Laguerre polynomials in
$\omega\left\vert x\right\vert ^{2}$ would provide bases for all~$G$-invariant wavefunctions.

\section{Appendix}

There is an analog $K\left(  x,y\right)  $ of the exponential function
$\exp\left\langle x,y\right\rangle $ on $\mathbb{R}^{N}\times\mathbb{R}^{N}$
which satisfies $K\left(  x,y\right)  =K\left(  y,x\right)  ,K\left(
xw,yw\right)  =K\left(  x,y\right)  $ for all $w\in W\left(  R\right)  $ and
$\mathcal{D}_{i}^{\left(  x\right)  }K\left(  x,y\right)  =y_{i}K\left(
x,y\right)  $ (where $\mathcal{D}_{i}^{\left(  x\right)  }$ is the operator
$\mathcal{D}_{i}$ acting on $x$, for $1\leq i\leq N$). The kernel exists for
nonsingular parameters $\left\{  \kappa_{v}\right\}  $, which include the
situation $\kappa_{v}\geq0$. Suppose $p\left(  x\right)  $ is a polynomial
then by the product rule%
\begin{align*}
\mathcal{D}_{i}\left(  p\left(  x\right)  K\left(  x,y\right)  \right)   &
=\left(  y_{i}p\left(  x\right)  +\frac{\partial}{\partial x_{i}}p\left(
x\right)  \right)  K\left(  x,y\right) \\
&  +\sum_{v\in R_{+}}\kappa_{v}\frac{p\left(  x\right)  -p\left(  x\sigma
_{v}\right)  }{\left\langle x,v\right\rangle }K\left(  x\sigma_{v},y\right)
v_{i}.
\end{align*}
This formula and the relation $wK\left(  x,y\right)  =K\left(  xw,y\right)
=K\left(  x,yw^{-1}\right)  $ show how an element of the rational Cherednik
algebra (the algebra of operators on polynomials generated by $\left\{
\mathcal{D}_{i}^{\left(  x\right)  },x_{i}:1\leq i\leq N\right\}  \cup
W\left(  R\right)  $) acts on a generic sum $\sum_{w\in W\left(  R\right)
}p_{W}\left(  x,y\right)  K\left(  xw,y\right)  $. It can be shown that if
$\mathcal{T}$ is in the rational Cherednik algebra and $\mathcal{T}K\left(
x,y\right)  =0$ then $\mathcal{T}=0$ (see Dunkl \cite{D1999}). For particular
groups and operators the calculation of $\mathcal{T}K\left(  x,y\right)  $ can
be implemented in computer algebra. The function $K$ is an undefined function
with argument $\left\langle x,y\right\rangle $ (or $\left\langle
x,yw^{-1}\right\rangle ).$ To compute $\mathcal{D}_{i}^{\left(  x\right)
}K\left(  xw,y\right)  =\mathcal{D}_{i}^{\left(  x\right)  }K\left(
x,yw^{-1}\right)  =\left(  yw^{-1}\right)  _{i}K\left(  xw,y\right)  $ one
applies $\frac{\partial}{\partial x_{i}}$ to $\left\langle x,yw^{-1}%
\right\rangle ,$ a straightforward calculation. This method was used to prove
(\ref{H2action}). Such calculations could involve as many as $\#G$ terms of
$K\left(  xw,y\right)  $ form.

To prove that $\left[  \widetilde{H}^{\left(  3\right)  },\mathcal{J}\right]
\neq0$ it suffices to prove this relation for $\kappa=0$. For a symbolic
argument analogous to the previous one replace $\mathcal{D}_{i}$ by
$\frac{\partial}{\partial x_{i}}+y_{i}$ and apply to polynomials in $x$. The
factor $\exp\left\langle x,y\right\rangle $ is understood; for example
$\mathcal{J}1=\left\vert x\right\vert ^{2}\left\vert y\right\vert
^{2}-\left\langle x,y\right\rangle \left(  \left\langle x,y\right\rangle
+2\right)  $ which manifests formula (\ref{angsq}) for $\mathcal{J}$ when
$\kappa=0$ and $N=3$ (since $\sum_{i=1}^{3}x_{i}\left(  \frac{\partial
}{\partial x_{i}}+y_{i}\right)  \left\langle x,y\right\rangle =\left\langle
x,y\right\rangle \left(  \left\langle x,y\right\rangle +1\right)  $.)


\begin{thebibliography}{99}                                                                                               %
\bibitem {HSC1973}H. S. M. Coxeter, \textit{Regular Polytopes}, 3rd ed'n,
Dover Press, New York, 1973.

\bibitem {D1989}C. F. Dunkl, Differential-difference operators associated to
reflection groups, \textit{Trans. Amer. Math. Soc. }\textbf{311}
(1989),\textbf{ }167-183.

\bibitem {D1999}C. F. Dunkl, Computing with differential-difference operators,
\textit{J. Symb. Comp.} \textbf{28} (1999), 819-826.

\bibitem {D2022}C. F. Dunkl, The B2 harmonic oscillator with reflections and
superintegrability, arXiv:2210.14180, 25 Oct 2022.

\bibitem {DX2014}C. F. Dunkl and Y. Xu, \textit{Orthogonal Polynomials of
Several Variables}, 2nd ed'n, Encyc. of Math. and its Applications 155,
Cambridge University Press, Cambridge, U.K., 2014.

\bibitem {E2010}P. Etingof, A uniform proof of the Macdonald-Mehta-Opdam
identity for finite Coxeter groups, \textit{Math. Res. Lett.} \textbf{17}
(2010), no. 2, 277--284.

\bibitem {FH2015}M. Feigin and T. Hakobyan, On Dunkl angular momenta algebra,
\textit{J. High Energy Physics} \textbf{11} (2015), 107.

\bibitem {G1989}F. G. Garvan, Some Macdonald-Mehta integrals by brute force,
$q$-\textit{series and partitions} (Minneapolis, MN, 1988), 77--98,
\textit{IMA Vol. Math. Appl.}, \textbf{18}, Springer, New York, 1989.

\bibitem {GV2014}V. X. Genest and L, Vinet, The multivariate Hahn polynomials
and the singular oscillator, \textit{J. Phys. A: Math. Theor.} \textbf{47}
(2014), 455201.

\bibitem {GVZ2014}V. X. Genest, L. Vinet, and A. Zhedanov, The Dunkl
oscillator in three dimensions, \textit{J. Phys.: Conf. Ser.} \textbf{512}
(2014) 012010.

\bibitem {LV1996}L. Lapointe and L. Vinet, Exact operator solution of the
Calogero-Sutherland model, \textit{Comm. Math. Phys. }\textbf{178}, (1996), 425-452.

\bibitem {Q2010}C. Quesne, Exchange operator formalism for an infinite family
of solvable and integrable quantum systems on a plane, \textit{Modern Phys.
Lett. \textbf{A} 25 }(2010), 15-24.

\bibitem {S1926}E. Schr\"{o}dinger, Quantisierung als Eigenwertproblem II,
\textit{Annalen der Physik}, (4) \textbf{79 }(1926), 489-527.

\bibitem {TTW2009}F. Tremblay, A. Turbiner, P. Winternitz, An infinite family
of solvable and integrable quantum systems on a plane, \textit{J. Phys. A:
Math. Theor.} \textbf{42}, (2009), 242001.
\end{thebibliography}
\end{document}